\documentclass[11pt]{article}

\usepackage[margin=1in]{geometry}
\DeclareMathAlphabet{\mathbbold}{U}{bbold}{m}{n}
\usepackage{amsthm,amsmath,amsfonts,amssymb}
\usepackage{mathtools}
\mathtoolsset{centercolon}

\usepackage{enumitem}

\usepackage{braket} 

\usepackage[pagebackref]{hyperref} 
\hypersetup{
    bookmarksnumbered=true, 
    unicode=false, 
    pdfstartview={FitH}, 
    pdftitle={Quantum linear systems algorithm with exponentially improved dependence on precision}, 
    pdfauthor={Andrew M. Childs, Robin Kothari, and Rolando D. Somma}, 
    pdfsubject={}, 
    pdfcreator={}, 
    pdfproducer={}, 
    pdfkeywords={}, 
    pdfnewwindow=true, 
    colorlinks=true, 
    linkcolor=blue, 
    citecolor=blue, 
    filecolor=blue, 
    urlcolor=blue 
}

\renewcommand{\backref}[1]{}

\renewcommand{\backrefalt}[4]{%
\ifcase #1 %
\or 
[p.\ #2]%
\else 
[pp.\ #2]%
\fi}

\makeatletter
\renewcommand{\paragraph}{%
  \@startsection{paragraph}{4}%
  {\z@}{2.25ex \@plus .5ex \@minus .3ex}{-1em}%
  {\normalfont\normalsize\bfseries}%
}
\makeatother

\newtheorem{theorem}{Theorem}
\newtheorem{lemma}[theorem]{Lemma}
\newtheorem{proposition}[theorem]{Proposition}
\newtheorem{corollary}[theorem]{Corollary}
\newtheorem{definition}[theorem]{Definition}

\theoremstyle{definition}
\newtheorem{problem}[theorem]{Problem}

\newcommand{\eq}[1]{\hyperref[eq:#1]{(\ref*{eq:#1})}}
\renewcommand{\sec}[1]{\hyperref[sec:#1]{Section~\ref*{sec:#1}}}
\newcommand{\thm}[1]{\hyperref[thm:#1]{Theorem~\ref*{thm:#1}}}
\newcommand{\lem}[1]{\hyperref[lem:#1]{Lemma~\ref*{lem:#1}}}
\newcommand{\prop}[1]{\hyperref[prop:#1]{Proposition~\ref*{prop:#1}}}
\newcommand{\cor}[1]{\hyperref[cor:#1]{Corollary~\ref*{cor:#1}}}
\newcommand{\fig}[1]{\hyperref[fig:#1]{Figure~\ref*{fig:#1}}}
\newcommand{\tab}[1]{\hyperref[tab:#1]{Table~\ref*{tab:#1}}}
\newcommand{\alg}[1]{\hyperref[alg:#1]{Algorithm~\ref*{alg:#1}}}
\newcommand{\app}[1]{\hyperref[app:#1]{Appendix~\ref*{app:#1}}}
\newcommand{\step}[1]{\hyperref[step:#1]{step~\ref*{step:#1}}}
\newcommand{\Step}[1]{\hyperref[step:#1]{Step~\ref*{step:#1}}}

\newcommand{\comment}[1]{}

\newcommand{\C}{{\mathbb{C}}}
\newcommand{\R}{{\mathbb{R}}}
\renewcommand{\O}{\widetilde{O}}

\newcommand{\A}{\mathcal{A}}
\renewcommand{\H}{\mathcal{H}}
\newcommand{\id}{\mathbbold{1}} 
\renewcommand{\varepsilon}{\delta}
\renewcommand{\d}{\mathrm{d}}

\DeclareMathOperator{\poly}{poly}
\renewcommand{\th}[1]{${#1}^{\textrm{th}}$}
\DeclareMathOperator{\GPE}{GPE}

\newcommand{\ceil}[1]{\lceil{#1}\rceil}
\newcommand{\floor}[1]{\lfloor{#1}\rfloor}

\newcommand{\norm}[1]{\|{#1}\|}
\newcommand{\Norm}[1]{\left\|{#1}\right\|}

\renewcommand{\(}{\left(}
\renewcommand{\)}{\right)}
\renewcommand{\>}{\rangle}
\newcommand{\<}{\langle}

\newcommand{\be}{\begin{equation}}
\newcommand{\ee}{\end{equation}}
\def\ba#1\ea{\begin{align}#1\end{align}}

\newcommand{\chebt}{\mathcal{T}}
\newcommand{\chebu}{\mathcal{U}}
\DeclareMathOperator{\spn}{span}
\DeclareMathOperator{\sgn}{sgn}

\newcommand{\boxA}{\mathcal{P}_A}
\newcommand{\boxB}{\mathcal{P}_B}

\newcommand{\cc}[1]{\mathsf{#1}} 

\newcommand{\mA}{{\mathcal A}}

\begin{document}

\title{Quantum algorithm for systems of linear equations \\ with exponentially improved dependence on precision}

\author{
Andrew M.\ Childs\thanks{
Department of Computer Science, Institute for Advanced Computer Studies,
and Joint Center for Quantum Information and Computer Science,
University of Maryland.
\texttt{amchilds@umd.edu}
} 
\and
Robin Kothari\thanks{
Center for Theoretical Physics, Massachusetts Institute of Technology.
\texttt{rkothari@mit.edu}
} 
\and
Rolando D.\ Somma\thanks{
Theoretical Division, Los Alamos National Laboratory.
\texttt{somma@lanl.gov}
} 
}

\date{}
\maketitle

\begin{abstract}
Harrow, Hassidim, and Lloyd showed that for a suitably specified $N \times N$ matrix $A$ and $N$-dimensional vector $\vec{b}$, there is a quantum algorithm that outputs a quantum state proportional to the solution of the linear system of equations $A\vec{x}=\vec{b}$. If $A$ is sparse and well-conditioned, their algorithm runs in time $\poly(\log N, 1/\epsilon)$, where $\epsilon$ is the desired precision in the output state. We improve this to an algorithm whose running time is polynomial in $\log(1/\epsilon)$, exponentially improving the dependence on precision while keeping essentially the same dependence on other parameters. Our algorithm 
is based on a general technique for implementing any operator with a suitable Fourier or Chebyshev series representation. This allows us to bypass the quantum phase estimation algorithm, whose dependence on $\epsilon$ is prohibitive.
\end{abstract}

\section{Introduction}
\label{sec:intro}

Recently, Harrow, Hassidim, and Lloyd \cite{HHL09} gave an efficient quantum algorithm for the Quantum Linear Systems Problem (QLSP). Here the goal is to prepare a quantum state $|x\>$ proportional to the solution $\vec{x}$ of a linear system of equations $A\vec{x} = \vec{b}$, given procedures for computing the entries of $A$ and for preparing a quantum state $|b\>$ proportional to $\vec{b}$. 
If the $N \times N$ matrix $A$ is $d$-sparse and has condition number $\kappa$, and if the procedures for computing entries of $A$ and for preparing $|b\>$ are efficient, then the Harrow--Hassidim--Lloyd (HHL) algorithm produces an $\epsilon$-approximation to the desired quantum state using $\poly(\log N, 1/\epsilon, d, 1/\kappa)$ resources (where the notation $\poly$ denotes a function upper bounded by a polynomial in its arguments). Note that the QLSP is not the same as the traditional problem of solving a linear system of equations \cite{Chi09,Aar15}.

The core of the HHL algorithm is an efficient procedure for simulating the dynamics of quantum systems.  Whereas straightforward approaches to quantum simulation using product formulas have complexity polynomial in $1/\epsilon$, recent work has given an algorithm with complexity $\poly(\log(1/\epsilon))$ \cite{BCC+14}, an exponential improvement in the dependence on $\epsilon$.  However, the performance of the HHL algorithm is limited by its use of phase estimation, which requires $\Omega(1/\epsilon)$ uses of a unitary operation to estimate its eigenvalues to precision $\epsilon$.  Thus, simply replacing the Hamiltonian simulation subroutine of the HHL algorithm with the best known method gives only a modest improvement, and in particular, still gives complexity $\poly(1/\epsilon)$.

In this work, we show how to circumvent the limitations of phase estimation, giving an algorithm for the QLSP that uses ideas from recent quantum simulation algorithms to apply the inverse of a matrix directly.  Under the same assumptions as for the HHL algorithm, our algorithm uses $\poly(\log N, \log(1/\epsilon), d, 1/\kappa)$ resources, exponentially improving the dependence on the precision parameter. 

To obtain this improvement, it is essential to consider QLSP as an inherently quantum problem, where the goal is to output a quantum state $|x\>$.  Originally, the HHL algorithm was described as a method for sampling expectation values of $|x\>$, providing a classical output \cite{HHL09}.  For the expectation value version of the problem, sampling error alone rules out the possibility of an algorithm with complexity $\poly(\log(1/\epsilon))$, unless $\cc{BQP}=\cc{PP}$ \cite[Theorem 6]{HHL09}.  However, the HHL algorithm actually solves the more general problem of outputting $|x\>$ (and the algorithm is commonly described in those terms \cite{Har15}). 

The improved performance of our approach may be especially useful when the quantum linear systems algorithm is used as a subroutine polynomially many times, so that its output must have inverse polynomial precision to guarantee that the final algorithm succeeds with high probability. An algorithm with $\poly(1/\epsilon)$ scaling incurs a polynomial overhead in running time due to error reduction, whereas an algorithm with $\poly(\log(1/\epsilon))$ scaling incurs only logarithmic overhead.

In fact, the results of this paper have already found applications. A recent work on speeding up the finite element method using quantum algorithms \cite{MP15} finds that quantum algorithms outperform classical algorithms only when the spatial dimension (of the partial differential equation to be solved) is larger than some threshold value. Their quantum algorithm uses our algorithm as a subroutine and would be worse by a polynomial factor if they used the previous best algorithm, which likely reduces the threshold value at which the quantum algorithm is superior.

Another recent algorithm to estimate hitting times of Markov chains \cite{CS16} uses the framework laid out in this paper (in \sec{LCU}) and closely follows our first approach, which we call the Fourier approach (\sec{Fourier}).

The improved scaling with $\epsilon$ may also find complexity-theoretic applications. A recent result on the power of $\mathsf{QMA}$ with exponentially small soundness--completeness gap \cite{FL16} crucially relies on the fact that the best Hamiltonian simulation algorithms have error dependence $\poly(\log(1/\epsilon))$.

\subsection{Problem statement}

The QLSP can be stated more precisely as follows.  We are given an $N \times N$ Hermitian matrix $A$ and a vector $\vec{b}=(b_1,b_2,\ldots,b_N)$.\footnote{The assumption that $A$ is Hermitian can be dropped without loss of generality \cite{HHL09}, since we can instead solve the linear system of equations given by $\(\begin{smallmatrix} 0&A\\ A^\dag &0 \end{smallmatrix}\)\vec{y}=\(\begin{smallmatrix} \vec{b}\\ 0 \end{smallmatrix}\)$, which has the unique  solution $\vec{y}=\(\begin{smallmatrix} 0\\ \vec{x} \end{smallmatrix}\)$ when $A$ is invertible. This transformation does not change the condition number of $A$. However, we need oracle access to the nonzero entries of the rows \emph{and} columns of $A$ when $A$ is not Hermitian.} 
The problem is to create the quantum state $|x\> :=  {\sum_{i=1}^N x_i |i\>}/\allowbreak{\norm{\sum_i x_i |i\>}}$, where the vector $\vec{x}=(x_1,x_2,\ldots,x_N)$ is defined by the equation $A\vec{x}=\vec{b}$. To obtain an algorithm running in time $\poly(\log N)$, we require succinct representations of $A$ and $\vec{b}$. As in \cite{HHL09}, we assume that access to $A$ and $\vec{b}$ is provided by black-box subroutines. For the vector $\vec{b}$, we assume there is a procedure $\boxB$ that produces the quantum state $|b\> := {\sum_i b_i |i\>}/{\norm{\sum_i b_i |i\>}}$. For the matrix $A$, we assume there is a procedure $\boxA$ that computes the locations and values of the nonzero entries. Specifically, as in the best known algorithm for Hamiltonian simulation \cite{BCK15}, we assume $\boxA$ allows us to perform the map
\begin{align}
  |j,\ell\> \mapsto |j,\nu(j,\ell)\>
\label{eq:entryloc}
\end{align}
for any $j \in [N] := \{1,\ldots,N\}$ and $\ell \in [d]$, where $d$ is the maximum number of nonzero entries in any row or column, and $\nu\colon [N] \times [d] \to [N]$ computes the row index of the \th{\ell} nonzero entry of the \th{j} column.  The procedure $\boxA$ also allows us to perform the map
\begin{align}
  |j,k,z\> \mapsto |j,k,z \oplus A_{jk}\>
\label{eq:entryval}
\end{align}
for any $j,k \in [N]$, where the third register holds a bit string representing an entry of $A$.  We assume the entries of $A$ can be represented exactly (or to sufficiently high precision that any error can be neglected).

Note that for the map \eq{entryloc}, we assume that the locations of the nonzero entries of $A$ can be computed in place, as in previous work on Hamiltonian simulation \cite{BC12,BCK15}.  This is possible if we can efficiently compute both $(j,\ell) \mapsto \nu(j,\ell)$ and the reverse map $(j,\nu(j,\ell)) \mapsto \ell$, which is possible for typical implicitly specified matrices $A$.  Alternatively, if $\nu$ provides the nonzero entries in ascending order, we can compute the reverse map with only a $\log d$ overhead by binary search.  In the worst case, if the entries are unordered, there may be an additional factor of $O(\sqrt{d})$ using Grover's algorithm \cite{Gro96} to compute the reverse map.

While the HHL algorithm solves the QLSP for all such matrices  $A$, it is efficient only when $A$ is sparse and well-conditioned. (We discuss later how the sparsity assumption can be slightly relaxed.) An $N \times N$ matrix is called \emph{$d$-sparse} if it has at most $d$ nonzero entries in any row or column. We call it simply \emph{sparse} if $d=\poly(\log N)$.  We call a matrix well-conditioned if its condition number is $\poly(\log N)$, where the \emph{condition number} of a matrix is the ratio of the largest to smallest singular value, and undefined when the smallest singular value of $A$ is $0$ (i.e., when $A$ is not invertible). Since $A$ is Hermitian, its singular values and eigenvalues are equal in magnitude. (Note that ``well-conditioned'' is a very strong requirement on the condition number, since it requires the condition number to be exponentially smaller than the dimension of the matrix.)

It will be convenient to quantify the resource requirements of algorithms solving the QLSP using two measures. First, by \emph{query complexity} we mean the number of uses of the procedure $\boxA$ (treating this procedure as a black box). By \emph{gate complexity}, we mean the total number of $2$-qubit gates used in the algorithm. We say an algorithm is \emph{gate-efficient} if its gate complexity is larger than its query complexity only by logarithmic factors. Formally, an algorithm with query complexity $Q$ is gate-efficient if its gate complexity is $O(Q \poly(\log Q ,\log N))$. We will also use \emph{expected} query or gate complexity to refer to an algorithm's query or gate cost in expectation (in the sense of a Las Vegas algorithm).

Formally, we define the Quantum Linear Systems Problem as follows:

\begin{problem}[QLSP] 
Let $A$ be an $N \times N$ Hermitian matrix with known condition number $\kappa$, $\norm{A} = 1$, and at most $d$ nonzero entries in any row or column. Let $\vec{b}$ be an $N$-dimensional vector, and let $\vec{x}:= A^{-1} \vec{b}$. We define the quantum states $|b\>$ and $|x\>$ as
\be
|b\> := \frac{\sum_i b_i |i\>}{\norm{\sum_i b_i |i\>}} \quad \mathrm{and} 
\quad |x\> :=  \frac{\sum_i x_i |i\>}{\norm{\sum_i x_i |i\>}}.
\ee
Given access to a procedure $\boxA$ that computes entries of $A$ as described in equations \eq{entryloc} and \eq{entryval} and a procedure $\boxB$ that prepares the state $|b\>$ in time $O(\poly(\log N))$, the goal is to output a state $|\tilde{x}\>$ such that $\norm{|\tilde{x}\>-|x\>} \leq \epsilon$, succeeding with probability $\Omega(1)$ (say, at least $1/2$), with a flag indicating success.
\end{problem}

We assume the condition number is known since our algorithm depends on it explicitly. More generally, we can replace $\kappa$ by an upper bound on the condition number at the expense of a corresponding increase in the running time.

While we only demand success with bounded error, repeating the procedure until it is successful gives an algorithm that is always correct and whose expected running time is asymptotically the same.  Alternatively, by repeating the procedure $O(\log(1/\epsilon))$ times, we can give an algorithm that always outputs a state $\epsilon$-close to the desired one.  To achieve this, the running time is simply multiplied by a factor of $O(\log(1/\epsilon))$.

\subsection{Results}
\label{sec:results}

Harrow, Hassidim, and Lloyd \cite{HHL09} present an algorithm for the QLSP that is efficient when $A$ is sparse and well-conditioned (i.e., when $d$ and $\kappa$ are both $\poly(\log N)$). 

\begin{theorem}[HHL algorithm]
\label{thm:HHL}
The QLSP can be solved by a gate-efficient algorithm that makes $O\big(\frac{d\kappa^2}{\epsilon} \poly(\log(d\kappa/\epsilon))\big)$ queries to the oracle $\boxA$ and $O(d\kappa \poly(\log(d\kappa/\epsilon)))$ uses of $\boxB$. 
\end{theorem}

The stated complexity uses the best known results on Hamiltonian simulation \cite{BCK15}, improving the $d$-dependence compared to \cite{HHL09}.  More generally, the HHL algorithm also works assuming only the ability to efficiently solve the Hamiltonian simulation problem for $A$, i.e., to efficiently implement the unitary operation $\exp(-iAt)$, without having direct access to $\boxA$. In other words, the HHL algorithm uses Hamiltonian simulation for $A$ as a black box.

We improve the $\epsilon$-dependence of the HHL algorithm from $\poly(1/\epsilon)$ to $\poly(\log(1/\epsilon))$, keeping essentially the same dependence on the other parameters. We provide two algorithms for the QLSP, one based on decomposing an operator using its Fourier series and another based on a decomposition into Chebyshev polynomials. The Fourier approach has slightly worse dependence on $\log(1/\epsilon)$, but uses Hamiltonian simulation as a black box only and is therefore more generally applicable.

\begin{theorem}[Fourier approach]
\label{thm:Fourier}
The QLSP can be solved with $O(\kappa \sqrt{\log(\kappa/\epsilon)})$ uses of a Hamiltonian simulation algorithm that approximates $\exp(-iAt)$ for $t=O(\kappa\log(\kappa/\epsilon))$ with precision $O(\epsilon/(\kappa\sqrt{\log(\kappa/\epsilon)}))$.  Using the best known algorithm for Hamiltonian simulation \cite{BCK15}, this makes $O(d\kappa^2 \log^{2.5}(\kappa/\epsilon))$ queries to $\boxA$, makes $O(\kappa \sqrt{\log(\kappa/\epsilon)})$ uses of $\boxB$, and has gate complexity
$O(d \kappa^2 \log^{2.5}(\kappa/\epsilon) (\log N + \log^{2.5}(\kappa/\epsilon)) )$. 
\end{theorem}

The second approach uses the oracle for the entries of $A$ directly without using Hamiltonian simulation as a subroutine, achieving better dependence on $\epsilon$.

\begin{theorem}[Chebyshev approach]
\label{thm:Chebyshev}
The QLSP can be solved using $O(d\kappa^2 \log^{2}(d\kappa/\epsilon))$ queries to $\boxA$ and $ O(\kappa  \log(d\kappa /\epsilon))$ uses of $\boxB$, with gate complexity  $O(d\kappa^2 \log^{2}(d\kappa/\epsilon) (\log N + \log^{2.5}(d\kappa/\epsilon)))$.
\end{theorem}

Both our algorithms achieve $\poly(\log(1/\epsilon))$ dependence on error and have similar complexity up to logarithmic terms, but are incomparable.  The Fourier approach is more general, applying whenever the Hamiltonian $A$ can be efficiently simulated (even if it is not necessarily sparse), and has slightly better dependence on $d$.  
The Chebyshev approach is more efficient in its dependence on $\kappa$ and $\epsilon$, but applies only to sparse Hamiltonians.

Ambainis later improved the $\kappa$-dependence of the HHL algorithm from quadratic to nearly linear \cite{Amb12}, which is essentially optimal since the dependence on $\kappa$ cannot be made sublinear \cite{HHL09}.
Since Ambainis's approach crucially uses phase estimation, applying it to our algorithms would increase their $\epsilon$-dependence back to $\poly(1/\epsilon)$. By carefully modifying the technique, we simultaneously achieve nearly linear scaling in $\kappa$ and logarithmic dependence on $\epsilon$.

\begin{theorem}[Linear scaling in $\kappa$]
\label{thm:linear}
The QLSP can be solved by a gate-efficient algorithm that makes $O\big({d\kappa}\poly(\log(d\kappa/\epsilon))\big)$ queries to the oracles $\boxA$ and $\boxB$. 
\end{theorem}

Additionally, if $\vec{b}$ is known to lie in an invariant subspace of $A$ of condition number $\kappa'<\kappa$, then $\kappa$ can be replaced by $\kappa'$ in any of our upper bounds (as in the HHL algorithm). In particular, this means our algorithm works even when $A$ is not invertible as long as $\vec{b}$ is known to lie outside the null space of $A$, i.e., when $\vec{x}=A^{-1}\vec{b}$ is well defined. This property is useful when the matrix $A$ is rectangular, as the reduction in \cite{HHL09} from rectangular to square matrices produces noninvertible matrices.

\subsection{High-level overview}

We now provide a high-level overview of our approach. The QLSP is equivalent to applying the (non-unitary) operator $A^{-1}$ to the state $|b\>$. Our general strategy is to represent $A^{-1}$, the operator we would like to perform, as a linear combination of unitaries we know how to perform. We then show how such a representation allows us to implement $A^{-1}$.

Our technique for implementing linear combinations of unitary operations arises from previous work on Hamiltonian simulation \cite{BCC+14,Kot14,BCC+15,BCK15} (see also \cite{SOGKL02,CW12} for previous related approaches).
As an example, consider implementing the operator $M = U_0 + U_1$, where $U_0$ and $U_1$ are unitaries with known quantum circuits. To implement $M$ on a state $|\psi\>$, we start with the state $|+\>|\psi\>$, where $|\pm\> := \frac{1}{\sqrt{2}}(|0\> \pm |1\>)$. We then perform the unitary $|0\>\<0|\otimes U_0 + |1\>\<1| \otimes U_1$, giving the state $\frac{1}{\sqrt{2}}(|0\>U_0|\psi\>+|1\>U_1|\psi\>)$. If we measure the first qubit in the $\{|+\>,|-\>\}$ basis and obtain the $|+\>$ outcome, then we prepare a state proportional to $M|\psi\>$. If we have the ability to create multiple copies of $|\psi\>$ or reflect about $|\psi\>$, then we can create the output state with high probability by repeating this process until we get the desired measurement outcome or by using amplitude amplification. As we describe in \sec{LCU}, this strategy for implementing a linear combination of unitaries works more generally.

However, it is unclear how to decompose $A^{-1}$ as a linear combination of easy-to-implement unitaries. Such a decomposition depends on what unitaries are used as elementary building blocks. Our first choice is to use the unitaries $\exp(-iAt)$, which can be implemented using any Hamiltonian simulation algorithm. We then have to represent $A^{-1}$ as $\sum_j \alpha_j \exp(-iAt_j)$ for some coefficients $\alpha_j$ and evolution times $t_j$. Since both sides of the equation are diagonal in the same basis, this is equivalent to representing $x^{-1}$ as a linear combination of $\sum_j \alpha_j \exp(-ixt_j)$.  The representation only needs to be correct for $x\in D_\kappa := [-1,-1/\kappa] \cup [1/\kappa,1]$ since we know the eigenvalues of $A$ fall in that range. For this range of $x$, we show how to approximate $x^{-1}$ as a linear combination of these unitaries in \sec{Fourier}. Our strategy for doing this broadly comprises the following steps. Since $x^{-1}$ is unbounded at the origin, we first ``tame'' the function by multiplying it with a function that is close to $1$ in the domain we care about but $0$ near the origin, so that the overall function is bounded. We then perform the Fourier transform to obtain an integral over various $\exp(-ixt)$. After making some approximations, we discretize the integral to obtain a finite sum over $\exp(-ixt)$.

Instead of using the unitaries $e^{-iAt}$, our second approach uses the operators $\chebt_n(A/d)$ as its building blocks, where $d$ is the sparsity of $A$ and $\chebt_n$ is the \th{n} Chebyshev polynomial of the first kind. These operators can be efficiently implemented using quantum walks \cite{Chi10,BC12}. Now the problem is to represent $x^{-1}$ as a linear combination $\sum_n \alpha_n T_n(x/d)$. Our strategy for obtaining this representation is similar to the previous case. We first tame the function by making it bounded near the origin. We thereby obtain a function that is exactly representable as a linear combination of Chebyshev polynomials, and we approximate this function by dropping low-weight terms, as described in \sec{Chebyshev}.
Combining these decompositions with the linear combination of unitaries strategy outlined above yields our main results, \thm{Fourier} and \thm{Chebyshev}. 

Lastly, we show how to decrease the $\kappa$-dependence of our algorithms to nearly linear while retaining the desired $\poly(\log(1/\epsilon))$ dependence on $\epsilon$. 
This improvement uses the observation that the complexity of our algorithms (and the HHL algorithm) is a product of two terms that depend on the eigenvalue $\lambda$ of $A$. Whereas the product of the maximum values of these two terms is quadratic in $\kappa$, the maximum of the product is only linear in $\kappa$. By introducing a technique called variable-time amplitude amplification, Ambainis exploited this observation to improve the $\kappa$-dependence of the HHL algorithm \cite{Amb12}. Since our approach deliberately avoids phase estimation, it is not immediately clear how to invoke a strategy based on resolving the eigenvalues of $A$. Nevertheless, we show that a careful application of low-precision phase estimation suffices to resolve the eigenvalues into buckets of exponentially increasing size, which allows us to apply the appropriate approximation of $A^{-1}$ for a given eigenvalue range. We then apply variable-time amplitude amplification on this algorithm, followed by additional post-processing to remove information left behind by phase estimation. In \sec{VTAA} we establish \thm{linear} and show how the $\kappa$-dependence of both algorithms can be made nearly linear, just as Ambainis showed for the original HHL algorithm.

Note that although we use the linear combination of unitaries approach (\sec{LCU}) to implement $A^{-1}$ using Fourier or Chebyshev expansions, the approach can be used generally to implement any function of $A$ that can be expressed as a linear combination of easy-to-implement unitaries. The problem of decomposing a function as an approximate linear combination of other functions, such as polynomials or trigonometric polynomials, is well studied in the field of approximation theory \cite{Che82,SV13}, and techniques from that literature might be applied in future applications of our techniques.  However, to the best of our knowledge, our specific results on approximation of the inverse by a Fourier or Chebyshev series are novel.  The closest work we are aware of shows how to approximate $A^{-1}$ as a linear combination of operators $\exp(-At)$ \cite[Ch.~12]{SV13}.  That expansion does not appear to be useful for our purposes (even taking $t$ imaginary) since it includes terms with long evolution times and large coefficients, resulting in high complexity.

\section{Implementing a linear combination of unitaries}
\label{sec:LCU}

We start by showing in \sec{framework} how to implement an operation $M$ that can be expressed as a linear combination of implementable unitaries. We then explain in \sec{application} how this primitive can be applied to solve the QLSP, given a suitable decomposition of $A^{-1}$.

\subsection{Framework}
\label{sec:framework}

A technique for implementing linear combinations of unitaries was introduced in some recent Hamiltonian simulation algorithms \cite{BCC+15,BCK15}. Since quantum simulation is unitary, $M$ is (nearly) unitary in the simulation algorithms based on this technique.  Furthermore, in Hamiltonian simulation we only have one copy of the input state. Under these circumstances, one can use oblivious amplitude amplification \cite{BCC+14} to implement $M$. In the QLSP, we can create multiple copies of the input state $|b\>$, so we do not require a tool like oblivious amplitude amplification (which would not work anyway when $M$ is far from unitary).

More precisely, our technique for implementing a linear combination of unitary operations is as follows.  Let $M = \sum_{i} \alpha_i U_i$ be a linear combination of unitary matrices $U_i$ with $\alpha_i >0$. 
We assume $\alpha_i>0$ without loss of generality since a phase can be subsumed into $U_i$. We show how to implement $M$ probabilistically using unitary operations $U$ and $V$ defined as follows.  
The operation $U := \sum_i |i\>\<i| \otimes U_i$ implements $U_i$ conditioned on the value of a control register. The operation $V$ maps $|0^m\>$ to $\frac{1}{\sqrt{\alpha}}\sum_i \sqrt{\alpha_i}|i\>$, where $\alpha := \norm{\vec{\alpha}}_1 = \sum_i \alpha_i$. Then, as shown in \cite[Lemma 2.1]{Kot14}, we can implement $M$ in the following sense. 

\begin{lemma}
\label{lem:lincomb}
Let $M = \sum_{i} \alpha_i U_i$ be a linear combination of unitaries $U_i$ with $\alpha_i >0$ for all $i$. Let $V$ be any operator that satisfies  $V|0^m\> := \frac{1}{\sqrt{\alpha}}\sum_i \sqrt{\alpha_i}|i\>$, where $\alpha := \sum_i \alpha_i$.
Then $W := V^{\dag}UV$ satisfies
\be
W|0^m\>|\psi\> =\frac 1 \alpha |0^m\>M|\psi\> + |\Psi^\perp\>
\ee
for all states $|\psi\>$, where $U := \sum_i |i\>\<i| \otimes U_i$ and $(|0^{m}\>\<0^{m}| \otimes \id)|\Psi^\perp\>=0$.
\end{lemma}

The lemma can be generalized to the case where $M = \sum_i \alpha_i T_i$, where each operator $T_i$ is not necessarily unitary, but is a block of a unitary, i.e.,  there exists a $U_i$ for which $U_i|0^{t}\>|\phi\> = |0^{t}\>T_i|\phi\> + |\Phi^\perp\>$ for all states $|\phi\>$, where $t\geq 0$ is an integer and $|\Phi^\perp\>$ has no overlap on states with $|0^{t}\>$ in the first register. We consider this more general situation in \lem{LCU} below.

In \lem{lincomb}, the operator $W$ can be thought of  as a postselected or probabilistic implementation of $M$ in the sense that, if we measure the first $m$ qubits of $W|0^m\>|\psi\>$ and observe the output $|0^m\>$, the state of the second register is proportional to $M|\psi\>$.  This successful outcome occurs with probability $(\norm{M|\psi\>}/\alpha)^2$. 

In our application, since we have the ability to create copies of the input state, we can repeat this process $O((\alpha/\norm{M|\psi\>})^2)$ times until the measurement yields the desired outcome. Alternately, we can construct the state with high probability using amplitude amplification \cite{BHMT02}. Amplitude amplification requires us to reflect about the starting state $|\psi\>$, which in our application is $|b\>$. With two uses of the procedure $\boxB$ for preparing $|b\>$ (one performed in reverse), we can reflect about $|b\>$. Since amplitude amplification yields a quadratic speedup, we obtain the desired state after $O(\alpha / {\norm{M|\psi\>}})$ rounds in expectation. Combining amplitude amplification and a more general \lem{lincomb}, we get the following procedure for implementing (submatrices of) linear combinations of unitary operations.

\begin{lemma}[Non-Unitary LCU Lemma]
\label{lem:LCU}
Let $M = \sum_i \alpha_i T_i$ with $\alpha_i>0$ for some (not necessarily unitary) operators $\{T_i\}$.  Let $\{U_i\}$ be a set of unitaries such that 
\be 
\label{eq:T}
U_i|0^{t}\>|\phi\> = |0^{t}\>T_i|\phi\> + |\Phi_i^\perp\>
\ee 
for all states $|\phi\>$, where $t$ is a nonnegative integer and $(|0^{t}\>\<0^{t}| \otimes \id)|\Phi_i^\perp\> = 0$.  Given an algorithm $\boxB$ for creating a state $|b\>$, there is a quantum algorithm that exactly prepares the quantum state $M|b\>/\norm{M|b\>}$ with constant success probability making
$O(\alpha/\norm{M|b\>})$ uses of $\boxB$, $U$, and $V$ in expectation, where 
\be 
\quad U := \sum_i |i\>\<i| \otimes U_i, \quad V|0^m\> = \frac{1}{\sqrt{\alpha}}\sum_i \sqrt{\alpha_i}|i\>, \quad \mathrm{and} \quad \alpha := \sum_i \alpha_i,
\ee
and that outputs a bit indicating whether it was successful.
\end{lemma}

\begin{proof}
We start by implementing the linear combination $M' = \sum_i \alpha_i U_i$. Using \lem{lincomb}, we have 
\be 
W|0^m\>|\psi\> = \frac 1 \alpha |0^m\>M'|\psi\> + |\Psi^\perp\>,
\ee 
where $W=V^\dag UV$. Now consider the action of $W$ on $|\psi\> = |0^t\>|\phi\>$:
\begin{align}  
W|0^{m+t}\>|\phi\> &=\frac 1 \alpha |0^m\>\Big(\sum_i \alpha_i U_i\Big)|0^t\>|\phi\> + |\Psi^\perp\>\\
& = \frac 1 \alpha |0^m\>|0^t\> \Big(\sum_i \alpha_i T_i\Big)|\phi\> 
+ \frac 1 \alpha |0^m\>\Big(\sum_i \alpha_i |\Phi_i^\perp\>\Big)
+ |\Psi^\perp\>\\
& = \frac 1 \alpha |0^{m+t}\> M|\phi\> + |\Xi^\perp\>,
\end{align}
where $|\Xi^\perp\>$ satisfies $(|0^{m+t}\>\<0^{m+t}| \otimes \id)|\Xi^\perp\> = 0$.

Now since we want to create a state proportional to $M|b\>$, let us plug in $|\phi\> = |b\>$. For convenience,  let $r = m+t$, which gives
\be 
W|0^{r}\>|b\>  =\frac 1 \alpha |0^{r}\> M|b\> + |\Xi^\perp\> = \Big(\frac 1 \alpha\norm{M|b\>}\Big) |0^{r}\> \frac{M|b\>}{\norm{M|b\>}}+ |\Xi^\perp\>.
\ee 

In words, applying $W$ on $|0^{r}\>|b\>$ followed by a measurement on the first $r$ qubits will yield the desired state $M|b\>/\norm{M|b\>}$ with probability  $(\norm{M|b\>}/\alpha)^2$. 

Since $\boxB$ is an algorithm that creates the state $|b\>$, we can also use it to reflect about the state $|0^{r}\>|b\>$. Specifically, say $\boxB$ performs the map $\boxB|0^s\> = |b\>$, then the operator $ \boxB(\id - 2|0^{r+s}\>\<0^{r+s}|)\boxB^\dag$ reflects about the state $|0^{r}\>|b\>$.

Given an algorithm that creates  a desired state from an initial state with probability $q$, and the ability to reflect about the initial state, amplitude amplification \cite[Theorem 3]{BHMT02} is a procedure for creating the desired state with probability (say) $1/2$ that makes $O(1/\sqrt{q})$ uses of the algorithm and the reflection map in expectation. Note that amplitude amplification also works when the probability $q$ is unknown when we are concerned with expected costs, and hence we do not need an estimate of $\norm{M|b\>}$ before we begin. If we want an algorithm with a worst-case guarantee on cost, then we need to know an upper bound on the probability $q$. 

Using amplitude amplification we get an algorithm that creates the quantum state  $M|b\>/\norm{M|b\>}$ and makes 
$ O(\alpha / {\norm{M|b\>}})$ uses of $W=V^\dag UV$ and $\boxB$ in expectation.
\end{proof}

To be completely general, we have stated our results for an arbitrary set of unitaries $\{U_i\}$ and quantified costs in terms of uses of $U := \sum_i |i\>\<i| \otimes U_i$. When we apply these results, we will need to compute the complexity of implementing $U$ in terms of the complexities of implementing the $U_i$. The query complexity of implementing $U$ is precisely the maximum query complexity of implementing any of the $U_i$, which is optimal since $U$ cannot be easier to implement than any particular $U_i$. 

On the other hand, the gate complexity of $U$ is not easy to characterize in terms of the gate complexities of $U_i$.
In general, the different $U_i$ may be completely unrelated and the cost of implementing $U$ may be greater than the gate cost of the most expensive $U_i$. However, in our applications the matrices $\{U_i\}$ are related and are, in fact, powers of a single unitary $Y$. In this case the gate complexity of implementing $U$ behaves nicely, as we show below.

\begin{lemma}
\label{lem:c-U}
Let $U = \sum_{i=0}^N |i\>\<i|\otimes Y^i$, where $Y$ is a unitary with gate complexity $G$. 
Let the gate complexity of $Y^{2^{j}}$ be $G_j \leq 2^jG$.
Then the gate complexity of $U$ is $O(\sum_{j=0}^{\floor{\log N}} G_j)=O(NG)$. \end{lemma}
\begin{proof}
Let $n := \floor{\log N}$ and consider the unitary $Y^{2^j}$ for $j\in\{0,\ldots,n\}$, which has gate complexity $G_j$, which is at most ${2^j}G$. Hence the controlled version of this unitary, $\textrm{c-}Y^{2^j}$, controlled by a single qubit, has gate complexity $O(G_j)$. The unitary $U$ can then be implemented by a circuit that performs, for all $j\in\{0,\ldots,n\}$, the operation $\textrm{c-}Y^{2^j}$ on the second register controlled by the \th{j} qubit of the first register $|i\>$. If the first register is in state $|i\>$, the operation performed on the second register is exactly $Y^i$, due to the binary encoding of the integer $i$. The gate complexity of this circuit is $O(\sum_{j=0}^{\floor{\log N}} G_j)=O(\sum_{j=0}^{\floor{\log N}} 2^jG)=O(NG)$.
\end{proof}

We can implement gates of the form $U=\sum_{i=0}^N \sum_{j=0}^M |i, j\> \<i,j| \otimes Y^{ij}$ similarly. We use an additional register to compute the product of $i$ and $j$ and, conditional on this register, we apply $Y^{ij}$ as described in \lem{c-U}. Last, we uncompute the product in the additional register. The gate complexity in this case is $O(NMG)$.

\subsection{Application to QLSP}
\label{sec:application}

We now describe how this technique can be applied to solve the QLSP.  Suppose we can approximate $A^{-1}$ by a linear combination of operators $T_i$ that are either unitary or can be implemented by unitaries as in \eq{T}. Then we can implement this linear combination of operators, which is approximately $A^{-1}$, using \lem{LCU}. The following sections, \sec{Fourier} and \sec{Chebyshev}, establish that $A^{-1}$ can indeed be represented in such a way using two different choices for the operators $T_i$ that correspond to using a Fourier decomposition and a Chebyshev decomposition, respectively. 

For our application, we need to show that implementing an operator close to $A^{-1}$ yields a state close to the desired one. More precisely, we need to show that if two operators $C$ and $D$ are close, then the normalized states $C|\psi\>/\norm{C|\psi\>}$ and $D|\psi\>/\norm{D|\psi\>}$ are also close.

\begin{proposition}\label{prop:statesclose}
Let $C$ be a Hermitian operator with $\norm{C^{-1}} \le 1$ (i.e., the smallest eigenvalue of $C$ in absolute value is at least $1$) and let $D$ be an operator that satisfies $\norm{C-D} \leq \epsilon <1/2$. Then the states $|x\> := C|\psi\>/\norm{C|\psi\>}$ and $|\tilde{x}\> := D|\psi\>/\norm{D|\psi\>}$ satisfy $\norm{|x\> - |\tilde{x}\>} \leq 4 \epsilon$.
\end{proposition}

\begin{proof}
Without loss of generality we assume $\norm{|\psi\>}=1$. Using the triangle inequality, we get 
\be 
\label{eq:tri}
\norm{|x\> - |\tilde{x}\>} 
= \Norm{\frac{C|\psi\>}{\norm{C|\psi\>}} - \frac{D|\psi\>}{\norm{D|\psi\>}}}
\le \Norm{\frac{C|\psi\>}{\norm{C|\psi\>}} - \frac{C|\psi\>}{\norm{D|\psi\>}}} 
  + \Norm{\frac{C|\psi\>}{\norm{D|\psi\>}} - \frac{D|\psi\>}{\norm{D|\psi\>}}}.
\ee
Again using the triangle inequality, we have $\norm{C|\psi\>} \leq \norm{D|\psi\>} + \norm{(C-D)|\psi\>} \leq \norm{D|\psi\>} + \epsilon$, which yields
\be
\bigl|{\norm{D|\psi\>} - \norm{C|\psi\>}}\bigr| \leq \epsilon \qquad \mathrm{and} \qquad \norm{D|\psi\>} \geq \norm{C|\psi\>}-\epsilon\geq1-\epsilon,
\ee
where we used the fact that $\norm{C|\psi\>} \geq 1$ since the smallest eigenvalue of $C$ in absolute value is at least 1.
Then we can upper bound the first term of the right-hand side of \eq{tri} as follows: 
\be
\Norm{\frac{C|\psi\>}{\norm{C|\psi\>}} - \frac{C|\psi\>}{\norm{D|\psi\>}}} 
\leq \frac{\bigl|{\norm{D|\psi\>} - \norm{C|\psi\>}}\bigr|}{\norm{D|\psi\>}} \\
 \leq \frac{\epsilon}{\norm{D|\psi\>}} \leq \frac{\epsilon}{1-\epsilon} \leq 2\epsilon.
\ee
Analogously, we can bound the second term on the right-hand side of \eq{tri} as follows:
\be
\Norm{\frac{C|\psi\>}{\norm{D|\psi\>}} - \frac{D|\psi\>}{\norm{D|\psi\>}}}
\leq \frac{\bigl|{\norm{C|\psi\>} - \norm{D|\psi\>}}\bigr|}{\norm{D|\psi\>}} 
 \leq \frac{\epsilon}{\norm{D|\psi\>}} \leq \frac{\epsilon}{1-\epsilon} \leq 2\epsilon.
\ee
Since both terms are at most $2\epsilon$, we have $\norm{|x\> - |\tilde{x}\>} \leq 4\epsilon$.
\end{proof}

Finally, we state a corollary that captures how we apply \lem{LCU} to implement $A^{-1}$.  This allows us to focus on approximating $1/x$ by a suitable linear combination of other functions.
We say that functions $f$ and $g$ are $\epsilon$-close on a domain $D \subseteq \R$ if $|f(x)-g(x)| \le \epsilon$ for all $x \in D$.

\begin{corollary}\label{cor:lcufunc}
Let $A$ be a Hermitian operator with eigenvalues in a domain $D \subseteq \R$.  Suppose the function $f\colon D \to \R$ satisfies $|f(x)| \geq 1$ for all $x\in D$ and is $\epsilon$-close to $\sum_i \alpha_i T_i$ on $D$ for some $\epsilon \in (0,1/2)$, coefficients $\alpha_i>0$, and functions $T_i\colon D \to \C$.  Let $\{U_i\}$ be a set of unitaries such that
\be 
U_i|0^{t}\>|\phi\> = |0^{t}\>T_i(A)|\phi\> + |\Phi_i^\perp\>
\ee
for all states $|\phi\>$, where $t$ is a nonnegative integer and $(|0^{t}\>\<0^{t}| \otimes \id)|\Phi_i^\perp\> = 0$.  Given an algorithm $\boxB$ for creating a state $|b\>$, there is a quantum algorithm that prepares a quantum state $4\epsilon$-close to $f(A)|b\>/\norm{f(A)|b\>}$, succeeding with constant probability, that makes an expected
$O(\alpha/\norm{f(A)\ket{b}}) = O(\alpha)$ uses of $\boxB$, $U$, and $V$, where 
\be 
\quad U := \sum_i |i\>\<i| \otimes U_i, \quad V|0^m\> := \frac{1}{\sqrt{\alpha}}\sum_i \sqrt{\alpha_i}|i\>, \quad \mathrm{and} \quad \alpha := \sum_i \alpha_i,
\ee
and outputs a bit indicating whether it was successful.  Furthermore, this algorithm can be modified to make $O(\alpha)$ uses of $\boxB$, $U$, and $V$ in the worst case.
\end{corollary}

\begin{proof}
By \lem{LCU}, we can exactly prepare the state $\sum_i \alpha_i T_i(A)|b\>/\norm{\sum_i \alpha_i T_i(A)|b\>}$ with constant success probability and with the stated resource requirements.  
Since the functions $f$ and $\sum_i \alpha_i T_i$ are $\epsilon$-close on a domain that includes the spectrum of $A$, we have $\norm{f(A)-\sum_i \alpha_i T_i(A)} < \epsilon$.
Since $|f(x)| \geq 1$ for all $x\in D$, the smallest eigenvalue of $f(A)$ in absolute value is at least $1$ and hence by \prop{statesclose}, the output state is $4\epsilon$-close to $f(A)|b\>/\norm{f(A)|b\>}$, as claimed.
Furthermore, $|f(x)| \ge 1$ implies $\norm{f(A)\ket{b}} \ge 1$, so $\alpha/\norm{f(A)\ket{b}} = O(\alpha)$.
Since $\alpha$ is known, by running the algorithm (say) ten times longer than its expected running time, we obtain a bounded-error algorithm that makes $O(\alpha)$ uses of $\boxB$, $U$, and $V$ in the worst case.
\end{proof}

In the following two sections, we apply this lemma to the function $f(x)=1/x$ on the domain $D_\kappa := [-1,-1/\kappa] \cup [1/\kappa,1]$, on which $f(x)$ satisfies $|f(x)|\geq 1$.  Then we have $\norm{f(A)|b\>}^{-1} \leq 1$, so the number of uses of $\boxB$, $U$, and $V$ is $O(\alpha)$.

\section{Fourier approach}
\label{sec:Fourier}

We now describe the Fourier approach, which is based on an approximation of $1/A$ as a linear combination of unitaries $e^{-iAt_i}$, $t_i \in \mathbb R$.  These unitaries can be implemented using any Hamiltonian simulation method.  For sparse $A$, we use the method of \cite{BCK15}.  Our quantum algorithm for the QLSP then uses \cor{lcufunc} to prepare a quantum state that is $\epsilon$-close to $|x\>$.
To that end, we establish the following Fourier expansion of the function $1/x$ on the domain $D_\kappa$ (proved in \sec{Fourier_gate}, where we give the explicit algorithm).

\begin{lemma}
\label{lem:FS1}
Let the function $h(x)$ be defined as 
\begin{align}
\label{eq:discrete0}
  h(x) :=
  \frac{i}{\sqrt{2\pi}} \sum_{j=0}^{J-1} \Delta_y \sum_{k=-K}^K \Delta_z \, 
  z_k e^{-z_k^2/2} e^{-i x y_j z_k},
\end{align}
where $y_j := j\Delta_y$, $z_k := k\Delta_z$, 
 for some fixed $J=\Theta(\frac{\kappa}{\epsilon} \log(\kappa/\epsilon))$, $K=\Theta(\kappa \log(\kappa/\epsilon))$, $\Delta_y = \Theta(\epsilon/ \sqrt{\log(\kappa/\epsilon)})$, and $\Delta_z = \Theta((\kappa \sqrt{\log(\kappa/\epsilon)})^{-1})$. Then $h(x)$
is $\epsilon$-close to $1/x$ on the domain $D_\kappa$.
\end{lemma}

We start by showing how to represent $1/x$ as a weighted double integral of $e^{-i x t}$ and then approximate the integrals by finite sums.
Given any odd function $f\colon \R \to \R$ satisfying $\int_0^\infty \d{y} \, f(y)=1$, we have $1/x = \int_0^\infty \d{y} \, f(xy)$ for $x\neq 0$.
To achieve good performance, we would like both $f$ and its Fourier transform to decay rapidly.  Here we choose $f(y) = y e^{-y^2/2}$ (although other choices are possible). As can be established using the Gaussian integral
$\int_{-\infty}^\infty \d{x}\, e^{-(x+c)^2}=\sqrt{\pi}$ for all $c \in \C$ (which we use throughout this section), this function satisfies
\begin{align}
	f(y) = \frac{i}{\sqrt{2 \pi}} \int_{-\infty}^\infty \d{z} \, z e^{-z^2/2} e^{-i y z}
  \label{eq:fteigen}
\end{align}
(i.e., it is an eigenfunction of the Fourier transform of eigenvalue $-i$, where the Fourier transform of a function $F\colon \R \to \C$ is the function $\hat F\colon \R \to \C$ given by $\hat F(k) := \frac{1}{\sqrt{2\pi}} \int_{-\infty}^\infty \d{x} \, F(x) e^{-i k x}$), so
\begin{align}
\label{eq:FSapproach}
  \frac{1}{x} 
  = \frac{i}{\sqrt{2 \pi}} \int_0^\infty \d{y} \int_{-\infty}^\infty \d{z} \,
    z e^{-z^2/2} e^{-i x y z}.
\end{align}
While this integral representation suffices to determine the query complexity, we must give an approximation as a finite sum (as in \lem{FS1}) to provide an explicit algorithm.
We can approximate \eq{FSapproach} by a finite sum by restricting to a finite range of integration (as established in \lem{FSfirstapprox}) and discretizing the integral.

In the remainder of this section, we analyze this approximation and thereby establish \thm{Fourier}.  We consider query complexity in \sec{Fourier_query} and then analyze gate complexity in \sec{Fourier_gate}.

\subsection{Query complexity}
\label{sec:Fourier_query}

The query complexity of this approach is determined by the query complexity of $U$ in \cor{lcufunc} (since the operation $V$ requires no queries to implement).
In turn, the query complexity of $U$ depends on the simulation precision and evolution times that are required to obtain an $\epsilon$-approximation of $1/A$.  To determine this, it suffices to understand the error introduced by truncating \eq{FSapproach} to a finite range of integration.

\begin{lemma}
\label{lem:FSfirstapprox}
The function
\begin{align}
  g(x) := 
  \frac{i}{\sqrt{2 \pi}} \int_0^{y_J} \d{y} \int_{-z_K}^{z_K} \d{z} \, 
  z e^{-z^2/2} e^{-i x y z}
\end{align}
is $\epsilon$-close to $1/x$ on the domain $D_\kappa$ for some $y_J = \Theta(\kappa \sqrt{\log(\kappa/\epsilon)})$ and $z_K=\Theta(\sqrt{\log(\kappa/\epsilon)})$.
\end{lemma}

\begin{proof}
Performing the integral over $y$ first, which does not change the value by Fubini's theorem, we have 
\begin{align}
  g(x) = \frac{1}{\sqrt{2\pi}x} \int_{-z_K}^{z_K} \d{z} \, 
  e^{-z^2/2} (1-e^{- i x y_J z}),
\end{align}
where we used the identity $\int_{a}^b \d{y}\, e^{cy}= \frac{1}{c}(e^{cb}-e^{ca})$.
Therefore
\begin{align}
  \biggl|g(x) - \frac{1}{x}\biggr|
  &= \biggl|g(x) - 
     \frac{1}{\sqrt{2\pi}x} \int_{-\infty}^\infty \d{z} \, e^{-z^2/2}\biggr| \\
  &= \frac{1}{\sqrt{2\pi}|x|}\biggl|-\int_{-\infty}^{\infty} \d{z} \, e^{-z^2/2} e^{- i x y_J z} 
  + \biggl(\int_{-\infty}^{-z_K} + \int_{z_K}^{\infty}\biggr) \d{z} \, 
  e^{-z^2/2} (e^{- i x y_J z}+1) \biggr| \\
  &\le \frac{1}{\sqrt{2\pi}|x|} \biggl| \int_{-\infty}^{\infty} \d{z} \, e^{-z^2/2} e^{-i x y_J z} \biggr|
   + \frac{4}{\sqrt{2\pi}|x|} \int_{z_K}^\infty \d{z} \, e^{-z^2/2} \\
  &= \frac{1}{|x|} e^{-(x y_J)^2/2}
  + \frac{4}{\sqrt{2\pi}|x|} \int_{z_K}^\infty \d{z} \, e^{-z^2/2} \\
  &\le \frac{1}{|x|} e^{-(x y_J)^2/2} + \frac{2}{|x|} e^{-z_K^2/2} ,
\label{eq:FSfirstapprox}
\end{align}
where in the last step we used the bound
\begin{align}
	\frac{1}{\sqrt{2\pi}}\int_{z_K}^\infty \d{z} \, e^{-z^2/2} 
  \le \frac{1}{2} e^{-z_K^2/2},
	\label{eq:erfbound}
\end{align}
which follows from the upper bound $\int_{x}^\infty \d{t}\, e^{-t^2} \leq \frac{\sqrt{\pi}}{2}e^{-x^2}$~\cite[7.1.13]{AS64}.  Since $|x| \ge 1/\kappa$, there exist $y_J = \Theta(\kappa \sqrt{\log(\kappa/\epsilon)})$ and $z_K=\Theta(\sqrt{\log(\kappa/\epsilon)})$
such that \eq{FSfirstapprox} is at most $\epsilon$.
\end{proof}

To apply \cor{lcufunc}, we discretize this integral, approximating $g(x)$ by the function $h(x)$ defined in \lem{FS1}.
In particular, $y_J$ and $z_K$ are as in the definition of $g(x)$ in \lem{FSfirstapprox}.
By taking $\Delta_y$ and $\Delta_z$ sufficiently small, $h(x)$ can approximate $g(x)$ arbitrarily closely.  Since the query complexity resulting from \cor{lcufunc} does not depend on the number of terms in the linear combination, we can make the discretization error arbitrarily small and neglect its contribution.  This shows the query complexity part of \thm{Fourier}, as follows.

\begin{proof}[Proof of \thm{Fourier} (query complexity)]
We implement $h(A)$ using \cor{lcufunc}.  The $L_1$ norm of the coefficients of this linear combination is
\begin{align}
  \alpha 
  &= \frac{1}{\sqrt{2\pi}} \sum_{j=0}^{J-1} \Delta_y \sum_{k=-K}^K \Delta_z \, |z_k| e^{-z_k^2/2}
  = \Theta(y_J) ,
\end{align}
where we used the fact that, for $\Delta_z \ll 1$,
\begin{align}
  \sum_{k=-K}^K \Delta_z \, |z_k| e^{-z_k^2/2}
  \approx \int_{-z_K}^{z_K} \d{z} \, |z| e^{-z^2/2}
  \le 2 \int_{0}^{\infty} \d{z} \, z e^{-z^2/2}
  = 2
\end{align}
(here the error in the approximation is negligible since we can take $\Delta_z$ arbitrarily small without affecting the query complexity).  By \lem{FSfirstapprox}, we have $\alpha = O(\kappa \sqrt{\log(\kappa/\epsilon)})$.  The evolution times of the Hamiltonian simulations that appear in the linear combination are $t=O(y_J z_K) = O(\kappa \log(\kappa/\epsilon))$.  Since we invoke Hamiltonian simulation $\Theta(\alpha)$ times, the overall error is at most $\epsilon$ provided each simulation has error $\epsilon' = O(\epsilon/\alpha) = O(\epsilon/(\kappa\sqrt{\log(\kappa/\epsilon)}))$.

The overall query complexity is the number of uses of Hamiltonian simulation times the query complexity of each simulation. Using the Hamiltonian simulation algorithm in \cite[Lemma 9]{BCK15}, the query complexity of simulating a $d$-sparse Hamiltonian for time $t$ with error at most $\epsilon'$ is
$O(d \norm{A}_{\max} t \log(\norm{A} t/\epsilon')/\log\log(\norm{A} t/\epsilon')) = O(dt \log(t/\epsilon'))$, since for the QLSP, we have $\norm{A}_{\max} \le \norm{A} \le 1$.
Thus the total query complexity is
\begin{align}
O(\alpha dt \log(t/\epsilon')) = O\Bigl(\kappa \sqrt{\log(\kappa/\epsilon)} d \kappa \log(\kappa/\epsilon) \log\Bigl(\frac{\kappa^2}{\epsilon}\log^{1.5}(\kappa/\epsilon)\Bigr)\Bigr)= O(d \kappa^2 \log^{2.5}( \kappa/\epsilon)).
\end{align}

Finally, by \cor{lcufunc}, the number of uses of $\boxB$ is $O(\alpha) = O(\kappa \sqrt{\log(\kappa/\epsilon)})$.
\end{proof}

\subsection{Gate complexity}
\label{sec:Fourier_gate}

The gate complexity of this approach is given by the gate complexity of the unitaries in \cor{lcufunc}, namely quantum simulation and the state preparation map $V$, times the number of amplitude amplification steps. The cutoffs $J$ and $K$ (which determine $\Delta_y$ and $\Delta_z$) affect the gate complexity as they determine the number of coefficients appearing in the approximation of $1/A$ by $h(A)$.

Our analysis uses the following identity.

\begin{lemma}
\label{lem:psf}
Let $\omega \in \mathbb R$ and $\Delta_z>0$. Then
\begin{align}
\label{eq:psf}
\sum_{k=-\infty}^{\infty} e^{-(\omega +2 \pi k/\Delta_z)^2/2}= \frac{1}{\sqrt{2 \pi}} \sum_{k=-\infty}^{\infty} \Delta_z \, e^{-z_k^2/2} e^{-i \omega z_k}    
\end{align}
where $z_k := k \Delta_z$.
\end{lemma}

\begin{proof}
Equation \eq{psf} is a case of the Poisson summation formula (see for example \cite[Section 11.11]{HN01}), which  states that if $f\colon \R \to \C$ is a Schwartz function with Fourier transform $\hat{f}\colon \R \to \C$, then
\begin{align}
\label{eq:Poisson}
\sum_{k=-\infty}^\infty f(k) =\sqrt{2 \pi} \sum_{k=-\infty}^\infty \hat f(2 \pi k).
\end{align}
Here $f$ is called a Schwartz function if, for all nonnegative integers $m$ and $n$, $\sup_{x \in \R} |x^m \tfrac{\d^n}{\d{x^n}} f(x)|$ is at most some constant (that can depend on $m$ and $n$).

In our case, $f(x)=e^{-(\omega +2 \pi x/\Delta_z)^2/2}$, with the Fourier transform $\hat{f}(y) = \frac{\Delta_z}{2 \pi} e^{-(y \Delta_z)^2/8 \pi^2} e^{-i y \Delta_z \omega/2\pi}$.
Then the left-hand side of \eq{Poisson} coincides with the left-hand side of \eq{psf}, and it is easy to check that the right-hand side of \eq{Poisson} coincides with the right-hand side of \eq{psf}.  It is well known that a Gaussian is a Schwartz function, so the identity follows.
\end{proof}

Now we quantify how well $h(x)$ in \eq{discrete0} approximates $1/x$.

\begin{proof}[Proof of \lem{FS1}]
We have $\Delta_y=y_J/J=\Theta(\epsilon/\sqrt{\log(\kappa/\epsilon)})$ and $\Delta_z=z_K/K=\Theta(1/\kappa\sqrt{\log(\kappa/\epsilon)})$.

Performing the geometric sum over $j$, we have
\begin{align}
  h(x) = \frac{i \Delta_y}{\sqrt{2\pi}} \sum_{k=-K}^K \Delta_z \, z_k e^{-z_k^2/2}
  \frac{1- e^{-i x y_J z_k}}{1 - e^{-i x \Delta_y z_k}} .
\end{align}
Using the triangle inequality, we have the bound
\begin{align}
	\Bigl|h(x)-\frac{1}{x}\Bigr|
	&\le \biggl|\frac{i \Delta_y}{\sqrt{2\pi}} \sum_{k=-K}^K \Delta_z \, z_k e^{-z_k^2/2}
  \frac{1- e^{-i x y_J z_k}}{1 - e^{-i x \Delta_y z_k}} - \frac{1}{\sqrt{2\pi}x} \sum_{k=-K}^K \Delta_z \, e^{-z_k^2/2}
  (1- e^{-i x y_J z_k})\biggr| \label{eq:FS1a}\\
&\quad + \biggl|\frac{1}{\sqrt{2\pi}x} \sum_{k=-K}^K \Delta_z \, e^{-z_k^2/2}
  (1- e^{-i x y_J z_k}) - \frac{1}{\sqrt{2\pi}x} \sum_{k=-\infty}^\infty \Delta_z \, e^{-z_k^2/2}
  (1- e^{-i x y_J z_k})\biggr| \label{eq:FS1b}\\
&\quad + \biggl|\frac{1}{\sqrt{2\pi}x} \sum_{k=-\infty}^\infty \Delta_z \, e^{-z_k^2/2}
  (1- e^{-i x y_J z_k}) - \frac{1}{\sqrt{2\pi}x} \sum_{k=-\infty}^\infty \Delta_z \,
  e^{-z_k^2/2}\biggr| \label{eq:FS1c}\\
&\quad + \biggl|\frac{1}{\sqrt{2\pi}x} \sum_{k=-\infty}^\infty \Delta_z \,
  e^{-z_k^2/2} - \frac{1}{x}\biggr|. \label{eq:FS1d}
\end{align}
We show that each term on the right-hand side is $O(\epsilon)$.

Since
$|\frac{1}{1-e^{-ix}}-\frac{1}{ix}| < 1$ for all $x \in [-1,1]$ (as is easily verified by plotting the left-hand side) and since $|x \Delta_y z_K|  = O(\epsilon) < 1$ for sufficiently small $\epsilon$, the term on the right-hand side of \eq{FS1a} is upper bounded by
\begin{align}
  \sqrt{\frac{2}{\pi}} \Delta_y \sum_{k=-K}^K \Delta_z \, |z_k| e^{-z_k^2/2}
  &\le \sqrt{\frac{2}{\pi}} \Delta_y \sum_{k=-K}^K \Delta_z \, e^{-z_k^2/4} \label{eq:FS1a1} \\
  &\le 2\sqrt{\frac{2}{\pi}} \Delta_y \int_0^\infty \d{z} \, e^{-z^2/4} \\
  &= 2\sqrt{2} \Delta_y = O(\epsilon).
\end{align}
Here in \eq{FS1a1} we used the fact that $|x| < e^{x^2/4}$ for all $x \in \R$.

The term in \eq{FS1b} is
\begin{align}
  \biggl|\frac{1}{\sqrt{2\pi}x} \sum_{|k|>K} \Delta_z \,
  e^{-z_k^2/2} (1- e^{-i x y_J z_k})\biggr|
  &\le \sqrt{\frac{2}{\pi}} \frac{2}{|x|} \sum_{k=K+1}^\infty \Delta_z \, e^{-z_k^2/2} \\
  &\le \sqrt{\frac{2}{\pi}} \frac{2}{|x|} \int_{z_K}^\infty \d{z} \, e^{-z^2/2} \\
  &\le \frac{2}{|x|} e^{-z_K^2/2} = O(\epsilon)
\end{align}
where we upper bounded the sum of a decreasing sequence by the corresponding integral and applied \eq{erfbound}.

Using \lem{psf} with $\omega=xy_J$, the term in \eq{FS1c} is
\begin{align}
  \frac{1}{\sqrt{2\pi} |x|} 
  \sum_{k=-\infty}^\infty \Delta_z \, e^{-z_k^2/2} e^{-i x y_J z_k}
  &= \frac{1}{|x|} \sum_{k=-\infty}^\infty e^{-(x y_J  +2 \pi k/\Delta_z )^2/2}.
\end{align}
For $k \ne 0$, we have $|x y_J + 2 \pi k/\Delta_z| \ge |k|(2\pi/\Delta_z - y_J)$.  By choosing $K$ sufficiently large, we can ensure that $2\pi/\Delta_z \ge y_J$.
Thus we see that \eq{FS1c} is at most
\begin{align}
	\frac{1}{|x|} \biggl( e^{-(x y_J)^2/2} + 2 \sum_{k=1}^\infty e^{-k(2\pi/\Delta_z - y_J)^2} \biggr)
	&= \frac{1}{|x|} \biggl( e^{-(x y_J)^2/2} + \frac{2}{e^{(2\pi/\Delta_z - y_J)^2}-1} \biggr) = O(\epsilon).
\end{align}

Finally, by \lem{psf} with $\omega=0$, we have
\begin{align}
	\frac{1}{\sqrt{2\pi}} \sum_{k=-\infty}^{\infty} \Delta_z \, e^{-z_k^2/2}
  &= \sum_{k=-\infty}^{\infty} e^{-(2 \pi k/\Delta_z)^2/2} \\
  &\le 1 + 2 \sum_{k=1}^\infty e^{-2\pi^2 k/\Delta_z^2} \\
  &= 1 + \frac{2}{e^{2\pi^2/\Delta_z^2}-1} \\
  &= 1 + O(\epsilon \, e^{-\kappa^2 \log\kappa}),
\end{align}
so \eq{FS1d} is also $O(\epsilon)$, completing the proof.
\end{proof}

We can now determine the gate complexity of the Fourier approach to the QLSP.

\begin{proof}[Proof of \thm{Fourier} (gate complexity)]
To apply \cor{lcufunc}, we must implement a unitary $V$ that maps the $m$-qubit state $\ket{0^m}$, where $m = O(\log(JK))$, to 
\begin{align}
  \frac{1}{(2 \pi)^{\frac{1}{4}}\sqrt{\alpha}} \sum_{j=0}^{J-1} \sqrt{\Delta_y} \sum_{k=-K}^K \sqrt{\Delta_z |z_k|} e^{-z_k^2/2} |j,k\> 
= \frac{\sqrt{\Delta_y}\sqrt{\Delta_z}}{(2 \pi)^{\frac{1}{4}}\sqrt{\alpha}} \Biggl( \sum_{j=0}^{J-1}|j\> \Biggr) \otimes  \Biggl(\sum_{k=-K}^K \sqrt{|z_k|} e^{-z_k^2/2} |k\>\Biggr).
\end{align}
From the query complexity part of the proof of \thm{Fourier}, we have $\alpha = O(\kappa\sqrt{\log(\kappa/\epsilon)})$. The operation $V$ can be implemented in two steps, preparing the superpositions over $|j\>$ and $|k\>$ independently since this is a product state. First we use $O(\log(J))$ Hadamard gates to prepare $\sum_{j=0}^{J-1} |j\>/\sqrt{J}$ (assuming for simplicity that $J$ is a power of $2$). 
Then we use $O(K)$ gates to prepare the corresponding superposition over $|k\>$ \cite[Sec IV]{SBM06}. 
Using the values of $J$ and $K$ from \lem{FS1}, the gate complexity of $V$ is
$O(\log(J)+K) = O(\log (\kappa \log(\kappa/\epsilon)/\epsilon) + \kappa \log(\kappa/\epsilon)) = O(\kappa \log(\kappa/\epsilon))$.

\cor{lcufunc} also requires the unitary
\begin{align}
  U = i \sum_{j=0}^{J-1} \sum_{k=-K}^K |j,k\>\<j,k| \otimes \sgn(k) e^{-iA y_j z_k}.
\end{align}
We now use the approach of \lem{c-U} (specifically, in the form discussed immediately following the proof of the lemma) and the Hamiltonian simulation algorithm of \cite[Lemma 10]{BCK15} to implement $U$ within precision $\epsilon'$. We can implement $U$ if we can implement the unitary $\sum_{j=0}^{J-1} \sum_{k=-K}^K |j,k\>\<j,k| \otimes Y^{jk}$, where $Y=e^{-iA\Delta_y\Delta_z}$. In the notation of \lem{c-U}, the gate cost of implementing this unitary is at most the sum of the costs of implementing $Y^{2^{r}}$ for $r=0$ to $r=\log JK = O(\log(\kappa/\epsilon))$. Since we would like to approximate $U$ to error $\epsilon'$, it suffices to implement the unitaries $Y^{2^{r}}$ to error $\bar{\epsilon}=O(\epsilon'/\log JK)=O(\epsilon'/\log(\kappa/\epsilon))$.

From \cite[Lemma 10]{BCK15}, we know that the gate complexity of simulating a $d$-sparse Hamiltonian $A$ for time $t$ with error $\bar{\epsilon}$ is
\begin{equation}\label{eq:HS}
O\left((d \norm{A}_{\max} t) (\log N + \log^{2.5}(\norm{A}t/\bar{\epsilon})) ({\log(\norm{A} t/\bar{\epsilon})})\right),
\end{equation}
where we have dropped a $\log\log(\norm{A}t/\bar{\epsilon})$ term that appeared in the denominator.
Lemma 10 of \cite{BCK15} does not explicitly handle the case where the evolution time $t$ is so short that some of these expressions are smaller than 1. In our application, we would like to simulate powers of the unitary $Y=e^{-iA\Delta_y\Delta_z}$, where $\Delta_y\Delta_z = \Theta(\epsilon/\kappa)$, which is indeed a very short evolution time. For such short times, the expression in \eq{HS} should have the first term replaced by $(d\norm{A}_{\max}t+1)$, and all logarithms should be prevented from dropping below 1, i.e., we treat every logarithm as a maximum of its original expression and 1.

Now we want to compute the gate cost of implementing $Y$ and higher powers of $Y$. The largest power involved is $Y^{JK}$, where $JK = \frac{\kappa^2}{\epsilon} \log^2 (\kappa/\epsilon)$. Hence the longest time for which Hamiltonian simulation is performed is $\tau:= y_Jy_K =  O(\kappa \log(\kappa/\epsilon))$. Since we can always use $\tau$ as an upper bound for $t$, we see that the cost of simulating $A$ for any time $t\leq \tau$ to error $\bar{\epsilon}$ is at most $O((dt+1) (\log N + \log^{2.5}(\tau/\bar{\epsilon})) \log(\tau/\bar{\epsilon}))$, since for the QLSP, we have $\norm{A}_{\max} \le \norm{A} \le 1$.

Hence the sum of costs of implementing $Y^{2^{r}}$ for $r=0$ to $r=\log JK = O(\log(\kappa/\epsilon))$ can be computed by first computing the sum of $(dt+1)$ over the range, i.e., the sum $\sum_{r=0}^{\log JK} (d\Delta_y\Delta_z 2^r+1)$. This sum is $O(d y_J z_K + \log JK) = O(d\kappa \log(\kappa/\epsilon))$.

Thus the sum of gate complexities of $Y^{2^{r}}$ is $O(d\kappa \log(\kappa/\epsilon) (\log N + \log^{2.5}(\tau/\bar{\epsilon})) \log(\tau/\bar{\epsilon}))$. Substituting the values of $\tau$, $\bar{\epsilon}$ and using $\epsilon' \le \epsilon$, the operation $U$ has gate complexity $O(d\kappa\log^2(\kappa/\epsilon ')(\log N+\log^{2.5}(\kappa/\epsilon ')))$.  

Note that the gate complexity of $V$ is dominated by that of $U$. Since we invoke Hamiltonian simulation $O(\alpha)$ times, with $\alpha=O(\kappa \sqrt{\log(\kappa/\epsilon)})$ (see the proof of the query complexity part of \thm{Fourier}),
we can choose $\epsilon'=O(\epsilon/\alpha)$ for overall error at most $\epsilon$.
Thus the overall gate complexity of the method is
\begin{align}
  O(d \kappa^2 \log^{2.5}(\kappa/\epsilon) (\log N + \log^{2.5}(\kappa/\epsilon)) ) 
\end{align}
as claimed.
\end{proof}

\section{Chebyshev approach}
\label{sec:Chebyshev}

We now describe our second approach to the QLSP.  This approach uses a Chebyshev expansion to implement $A^{-1}$ without appealing directly to Hamiltonian simulation.  Instead of building the function $1/x$ as a linear combination of terms of the form $e^{-itx}$ (as in \sec{Fourier}), we approximate it by a linear combination of terms of the form $\chebt_n(x)$, where $\chebt_n$ is the \th{n} Chebyshev polynomial of the first kind.  We implement such terms using a quantum walk that has previously been applied to Hamiltonian simulation \cite{Chi10,BC12,BCK15}. With that goal, we establish the following decomposition of $1/x$ as a linear combination of Chebyshev polynomials.

\begin{lemma}
Let $g(x)$ be defined as 
\begin{equation}
  g(x) := 4 \sum_{j=0}^{j_0} (-1)^j \left[\frac{\sum_{i=j+1}^b \binom{2b}{b+i}}{2^{2b}}\right] \chebt_{2j+1}(x),
\end{equation}
where $j_0 = \sqrt{b\log(4b/\epsilon)}$ and $b = \kappa^2 \log(\kappa/\epsilon)$. Then $g(x)$ is $2\epsilon$-close to $1/x$ on $D_\kappa$.
\end{lemma}

The Chebyshev polynomials of the first kind are defined as follows: $\chebt_0(x) =1$, $\chebt_1(x) = x$, and $\chebt_{n+1}(x) = 2x\chebt_n(x)-\chebt_{n-1}(x)$. They are also the unique polynomials satisfying $\chebt_n(\cos \theta) = \cos n\theta$.

This section is organized as follows.  In \sec{Chebyshev_atoms} we define a quantum walk for any given Hamiltonian and express this walk (and its powers) in terms of Chebyshev polynomials.  In \sec{Chebyshev_query}, we present an expansion of $1/x$ in terms of Chebyshev polynomials and use this to establish the query complexity of the Chebyshev approach to the QLSP.  In \sec{Chebyshev_gate}, we upper bound the gate complexity of this approach.

\subsection{A quantum walk for any Hamiltonian}
\label{sec:Chebyshev_atoms}

Let $A$ be a $d$-sparse $N \times N$ Hamiltonian with $\norm{A}_{\max} \le 1$.  We now define a quantum walk corresponding to $A$ and express its action in terms of Chebyshev polynomials.

The quantum walk is defined using a set of states $\{|\psi_j\> \in \C^{2N} \otimes \C^{2N}: j \in [N]\}$ defined as 
\begin{align}
\label{eq:psi_j}
  |\psi_j\> := |j\> \otimes \frac{1}{\sqrt d} \sum_{k \in [N]: A_{jk}\neq 0} \Bigl(  \sqrt{A_{jk}^*}|k\> + \sqrt{1-|A_{jk}|} \, |k+N\>\Bigr).
\end{align}
Note that the square root in \eq{psi_j} is potentially ambiguous when $A_{jk}$ is complex. Our results hold for any consistent choice of square root that ensures $\sqrt{\smash[b]{A_{jk}^*}} \bigl(\sqrt{\smash[b]{A_{jk}}}\bigr)^* = A_{jk}^*$. For more detail, see the discussion preceding \cite[eq.~14]{BC12}. We assume in \eq{psi_j} that the oracle returns exactly $d$ nonzero entries for a given $j$. This is without loss of generality as we can modify the oracle to treat some zero entries as nonzero entries to make up the difference. On these additional values of $k$, the oracle will return the value of $A_{jk}$ to be $0$, but these entries will contribute to the state in \eq{psi_j} due to the term $\sqrt{1-|A_{jk}|} \, |k+N\>$.

The quantum walk now occurs in the Hilbert space $\C^{2N} \otimes \C^{2N}$.  We define an isometry $T$ from $\C^{N}$ to the walk space by
\begin{align}
  T := \sum_{j \in [N]} |\psi_j\>\<j|.
\end{align}
We define a swap operator on the walk space by $S|j,k\>=|k,j\>$ for all $j,k \in [2N]$.  Finally, the 
walk operator is $W := S(2TT^\dag - \id)$. 
As discussed in \cite{BC12,BCK15}, the walk operator $W$ can be implemented using $O(1)$ queries to $\boxA$, and this implementation is gate-efficient as we discuss in \sec{Chebyshev_gate}.
(For the details of the implementation, see Lemma 10 of \cite{BCK15}.  In that treatment, the states are defined using an ancilla qubit instead of $N$ additional basis states, but the convention in \eq{psi_j} is easily recovered by the efficient---and efficiently invertible---mapping from $\C^N \otimes \C^2 \to \C^{2N}$ defined by $|k,0\> \mapsto |k\>$ and $|k,1\> \mapsto |k+N\>$ for all $k\in [N]$.)

We now analyze the structure of this walk. It will be convenient to consider the matrix $H :=A/d$ in the following. Note that since $\norm{A}_{\max}\leq 1$ and $A$ is $d$-sparse, we have that $\norm{H}\leq 1$.

\begin{lemma}\label{lem:walk}
Let $\ket{\lambda}$ be an eigenvector of $H := A/d$ with eigenvalue $\lambda \in (-1,1)$.  Within the invariant subspace $\spn\{T|\lambda\>, ST|\lambda\>\}$, the walk operator $W$ has the block form
\[
  \begin{pmatrix}
	\lambda & -\sqrt{1-\lambda^2} \\
	\sqrt{1-\lambda^2} & \lambda
	\end{pmatrix}
\]
where the first row/column corresponds to the state $T|\lambda\>$. For an eigenvector $\ket{\lambda}$ with eigenvalue $|\lambda|=1$, we have $WT|\lambda\>=\lambda  T |\lambda\>$.
\end{lemma}

\begin{proof}
Observe that $T^\dag T = \id$ and $T^\dag S T = H$.
Let $|{\perp_\lambda}\>$ denote a state orthogonal to $T|\lambda\>$ in $\spn\{T|\lambda\>, ST|\lambda\>\}$, satisfying
\begin{align}\label{eq:perpdef}
	ST|\lambda\> = \lambda T|\lambda\> + \sqrt{1-\lambda^2} |{\perp_\lambda}\>,
\end{align}
which is well defined since $|\lambda|<1$. Then we have
\begin{align}
	WT|\lambda\>
	&= S(2TT^\dag - 1)T|\lambda\> \\
	&= ST|\lambda\> \label{eq:WTlameqSTlam} \\
	&= \lambda T|\lambda\> + \sqrt{1-\lambda^2}|{\perp_\lambda}\>.
	\label{eq:wtlambda}
\end{align}
Furthermore,
\begin{align}
	\sqrt{1-\lambda^2} W|{\perp_\lambda}\>
	&= WST|\lambda\> - \lambda WT|\lambda\> \\
	&= S(2TT^\dag-1)ST|\lambda\> - \lambda S(2TT^\dag-1)T|\lambda\> \\
	&= \lambda ST|\lambda\> - T|\lambda\> \\
	&= (\lambda^2-1) T|\lambda\> + \lambda \sqrt{1-\lambda^2} |{\perp_\lambda}\>,
\end{align}
so dividing by $\sqrt{1-\lambda^2}$, which is nonzero, we have
\begin{align}
	W|{\perp_\lambda}\> = -\sqrt{1-\lambda^2} T|\lambda\> + \lambda |{\perp_\lambda}\>.
	\label{eq:wperplambda}
\end{align}
Combining \eq{wtlambda} and \eq{wperplambda}, we see that $W$ has the stated form. When $|\lambda|=1$, \eq{perpdef} simplifies to $ST|\lambda\>=\lambda T|\lambda\>$, so \eq{WTlameqSTlam} gives $WT|\lambda\>=\lambda  T |\lambda\>$, as desired.
\end{proof}

\begin{lemma}\label{lem:power}
	For any $\lambda \in [-1,1]$, let
	\[
	  W =
	  \begin{pmatrix}
		\lambda & -\sqrt{1-\lambda^2} \\
		\sqrt{1-\lambda^2} & \lambda
		\end{pmatrix}.
	\]
	Then for any positive integer $n$,
	\[
	  W^n =
	  \begin{pmatrix}
    \chebt_n(\lambda) & -\sqrt{1-\lambda^2} \, \chebu_{n-1}(\lambda) \\
    \sqrt{1-\lambda^2} \, \chebu_{n-1}(\lambda) & \chebt_n(\lambda)
    \end{pmatrix}
	\]
where $\chebt_n$ is the \th{n} Chebyshev polynomial of the first kind and $\chebu_n$ is the \th{n} Chebyshev polynomial of the second kind.
\end{lemma}

\begin{proof}
This follows by straightforward induction.  The base case ($n=1$) follows from the identities $\chebt_1(\lambda) = \lambda$ and $\chebu_0(\lambda)=1$.  Assuming the claim holds for a given value of $n$, we have
\begin{align}
	W^n W
	&= \begin{pmatrix}
	   \lambda\,\chebt_n(\lambda) - (1-\lambda^2)\,\chebu_{n-1}(\lambda) & 
	   -\sqrt{1-\lambda^2}\,(\chebt_n(\lambda)+\lambda\,\chebu_{n-1}(\lambda)) \\
	   \sqrt{1-\lambda^2}\,(\lambda\,\chebu_{n-1}(\lambda)+\chebt_n(\lambda)) &
	   -(1-\lambda^2)\,\chebu_{n-1}(\lambda)+\lambda\,\chebt_n(\lambda)
     \end{pmatrix} \\
  &= \begin{pmatrix}
     \chebt_{n+1}(\lambda) & -\sqrt{1-\lambda^2} \, \chebu_n(\lambda) \\
     \sqrt{1-\lambda^2} \, \chebu_n(\lambda) & \chebt_{n+1}(\lambda)
     \end{pmatrix}
\end{align}
where we used the identities
\begin{align}
	\chebt_{n+1}(\lambda) 
	&= \lambda \, \chebt_n(\lambda) - (1-\lambda^2) \, \chebu_{n-1}(\lambda) \\
	\chebu_n(\lambda)
	&= \chebt_n(\lambda) + \lambda \, \chebu_{n-1}(\lambda)
\end{align}
relating Chebyshev polynomials of the first and second kinds. This lemma can alternately be proved using the substitution $\lambda = \cos \theta$ and the trignometric relations between Chebyshev polynomials.
\end{proof}

From \lem{walk} and \lem{power}, we get that for any eigenvector $|\lambda\>$ of $H$ with $|\lambda|<1$, 
\begin{equation}
W^n T|\lambda\> = \chebt_n(\lambda) T|\lambda\> +  \sqrt{1-\lambda^2} \, \chebu_{n-1}(\lambda) |{\perp_\lambda}\>,
\end{equation}
where $|{\perp_\lambda}\>$ is defined through \eq{perpdef}. This equation can also be extended to eigenvectors $|\lambda\>$ with $|\lambda|=1$ by observing that $\chebt_n(\lambda) = \lambda^n$ when $\lambda\in\{-1,+1\}$ and by defining $|{\perp_\lambda}\>=0$.

Since the eigenvectors $|\lambda\>$ of $H$ span all of $\C^N$, and the pairs $\{T|\lambda\>,|{\perp_\lambda}\>\}$ form an invariant subspace of $W$, we have that for any $|\psi\> \in \C^N$, 
\begin{equation}
W^n T|\psi\> = T \chebt_n(H)|\psi\> + |{\perp_\psi}\>,
\end{equation}
where $|{\perp_\psi}\>$ is an unnormalized state that is orthogonal to $\spn\{T|j\>:j\in[N]\}$.

Now consider a unitary circuit implementation of the isometry $T$. Since $T$ maps $\C^N$ to $\C^{2N} \times \C^{2N}$, a unitary implementation would map $|0^m\>|\psi\>$ to $T|\psi\>$ for any $|\psi\> \in \C^N$, with $m = \ceil{\log 2N} + 1$.
By applying this unitary, followed by $W^n$, followed by the inverse of this unitary, we can implement
\begin{equation}
\label{eq:ChebH}
|0^m\>|\psi\> \mapsto |0^m\> \chebt_n(H)|\psi\> + |\Phi^\perp\>,
\end{equation}
where $|\Phi^\perp\>$ is an unnormalized state satisfying $\Pi|\Phi^\perp\>=0$, where $\Pi := |0^{m}\>\<0^{m}| \otimes \id$.

Finally, we note that since the walk operator $W$ (and $T$) can be implemented using $O(1)$ queries to the oracle $\boxA$, as described in \cite{BC12,BCK15}, the operation in \eq{ChebH} can be performed using $O(n)$ queries.

\subsection{Query complexity}
\label{sec:Chebyshev_query}

The Hamiltonian simulation algorithm of \cite{BCK15} approximates $e^{-iHt}$ by a linear combination of the matrices $W^n$, where $W$ is the quantum walk defined above.  From the form of $W^n$ in terms of Chebyshev polynomials, we can interpret the algorithm as decomposing $e^{-iHt}$ as a linear combination of the Chebyshev polynomials $\chebt_n(H)$.

Applying a similar approach to linear systems, we aim to approximate $H^{-1}$, where $H = A/d$, by a linear combination of $\chebt_n(H)$. Since $H$ is Hermitian, it suffices to find a good approximation of $f(x)=1/x$ by a linear combination of $\chebt_n(x)$ over a domain containing the spectrum of $H$. Since the eigenvalues of $A$ lie in $D_\kappa := [-1,-1/\kappa]\cup[1/\kappa,1]$, the eigenvalues of $H$ lie in $D_{\kappa d}$. 
(Although the eigenvalues of $H$ actually lie in $[-1/d,-1/(\kappa d)]\cup[1/(\kappa d),1/d]$, we do not use this fact.)  Once we express $f(x)=1/x$ as a linear combination of $\chebt_n(x)$ over the domain $D_{\kappa d}$, the same linear combination yields an approximation for $H^{-1}$, according to \cor{lcufunc}.

We now analyze the approximation of $1/x$ by a Chebyshev series.
We start by approximating the function $1/x$ with a function that is bounded at the origin. We accomplish this by multiplying $1/x$ with a function that is close to identity on $D_{\kappa d}$ and very small near the origin. The function $1-(1-x^2)^b$ has this property for large enough $b$.

\begin{lemma}
\label{lem:taming}
The function 
\begin{equation}\label{eq:taming}
  f(x) := \frac{1-(1-x^2)^b}{x}
\end{equation}
is $\epsilon$-close to $1/x$ on the domain $D_{\kappa d}$ for any integer $b \geq (\kappa d)^2 \log(\kappa d/\epsilon)$.
\end{lemma}

\begin{proof}
For $b>0$, on the domain $D_{\kappa d}$ the numerator $1-(1-x^2)^b$ increases monotonically toward $1$ as $|x|$ increases.  Thus, over $D_{\kappa d}$, the numerator differs most from the constant function $1$ at $x=1/(\kappa d)$, where the difference is
\begin{equation}
\(1-\frac{1}{(\kappa d)^2}\)^b \le e^{-b/(\kappa d)^2} \leq \frac{\epsilon}{\kappa d}.
\end{equation}
Therefore we have
\be 
  \biggl|\frac{1-(1-x^2)^b}{x}-\frac{1}{x}\biggr| 
  \le \frac{\epsilon}{\kappa d |x|}
  \le \epsilon.
\qedhere
\ee 
\end{proof}

We now express the function $f(x)$ as a linear combination of Chebyshev polynomials. Since $f(x)$ is a polynomial of degree $2b-1$, it can be exactly represented as a linear combination of Chebyshev polynomials of order at most $2b-1$.

\begin{lemma}
\label{lem:exact}
Over the domain $[-1,1]$, the function $f(x)$ defined in \eq{taming} can be exactly represented by a linear combination of Chebyshev polynomials of order at most $2b-1$:
\begin{equation}
  f(x) = \frac{1-(1-x^2)^b}{x} = 4 \sum_{j=0}^{b-1} (-1)^j \left[\frac{\sum_{i=j+1}^b \binom{2b}{b+i}}{2^{2b}}\right] \chebt_{2j+1}(x). 
\label{eq:exact}
\end{equation}
\end{lemma}

\begin{proof}
First note that $f(x)$ is well-defined as $x \to 0$ because its numerator is a polynomial with no constant term.

Since $x\in[-1,1]$, we can make the substitution $x=\cos(\theta)$. Thus we aim to prove
\begin{equation}
  f(\cos(\theta))
  = \frac{1-(1-\cos^2(\theta))^b}{\cos(\theta)} 
  = 4 \sum_{j=0}^{b-1} (-1)^j \left[\frac{\sum_{i=j+1}^b \binom{2b}{b+i}}{2^{2b}}\right] \chebt_{2j+1}(\cos(\theta)).
\end{equation}
Using $\chebt_n(\cos \theta) = \cos(n\theta)$, this is equivalent to showing
\begin{equation}
{1-(\sin(\theta))^{2b}} = 4 \sum_{j=0}^{b-1} (-1)^j \left[\frac{\sum_{i=j+1}^b \binom{2b}{b+i}}{2^{2b}}\right] \cos((2j+1)\theta)\cos(\theta).
\label{eq:newid}
\end{equation}

By the identity
\begin{align}
  \sin(\theta)^{2b}
  &= \frac{1}{(2i)^{2b}} \sum_{k=0}^{2b} \binom{2b}{k} (e^{i \theta})^k (-e^{-i\theta})^{2b-k} \\
  &= \frac{(-1)^b}{2^{2b}} \sum_{k=0}^{2b} \binom{2b}{k} (-1)^k e^{i (2k-2b) \theta} \\
  &= \frac{1}{2^{2b}} \binom{2b}{b} + \frac{2}{2^{2b}} \sum_{k=0}^{b-1} \binom{2b}{k} (-1)^{b-k} \cos((2b-2k)\theta),
\end{align}
we have
\begin{align}
  {1-(\sin(\theta))^{2b}} = 1 - \( \frac{1}{2^{2b}} \binom{2b}{b} + \frac{2}{2^{2b}} \sum_{k=0}^{b-1} (-1)^{b-k} \binom{2b}{k}\cos((2b-2k)\theta) \).
\end{align}
Using the product-to-sum formula for cosines ($2\cos \theta \cos \phi = \cos(\theta+\phi)+\cos(\theta-\phi)$), the right-hand side of \eq{newid} simplifies to
\begin{align}
  2 \sum_{j=0}^{b-1} (-1)^j \left[\frac{\sum_{i=j+1}^b \binom{2b}{b+i}}{2^{2b}}\right] \(\cos((2j+2)\theta)+\cos(2j\theta)\).
\end{align}
It remains to check that the coefficients of $\cos(2r\theta)$ are equal on both sides for all $r \in \{0,1,\ldots,b\}$. (Note that both sides are linear combinations of such terms.)

First, the coefficient of the $r=0$ term on the left-hand side is $1 - \frac{1}{2^{2b}} \binom{2b}{b}$. The same coefficient on the right-hand side is $2^{1-2b} \sum_{i=1}^b \binom{2b}{b+i} = 2^{-2b} \(\sum_{i=0}^{2b} \binom{2b}{i}-\binom{2b}{b}\) = 1 - \frac{1}{2^{2b}} \binom{2b}{b}$.

When $r \in [b]$, the coefficient of $\cos(2r\theta)$ on the left-hand side is
\begin{equation}
 - \( \frac{2}{2^{2b}} (-1)^{r} \binom{2b}{b-r} \).
\end{equation}

When $r \in [b-1]$, the coefficient of $\cos(2r\theta)$ on the right-hand side is
\begin{equation}
  2 (-1)^{r-1} \left[\frac{\sum_{i=r}^b \binom{2b}{b+i}}{2^{2b}}\right] -
  2 (-1)^r \left[\frac{\sum_{i=r+1}^b \binom{2b}{b+i}}{2^{2b}}\right] 
  = 2(-1)^{r-1} \left[ \frac{\binom{2b}{b+r}}{2^{2b}}\right],
\end{equation}
which is the same as the left-hand side.

Finally, when $r=b$, the coefficient of $\cos(2r\theta)$ on the right-hand side is
\begin{equation}
  2 (-1)^{b-1} \left[\frac{\sum_{i=b}^b \binom{2b}{b+i}}{2^{2b}}\right] = 2 (-1)^{b-1} \left[\frac{1}{2^{2b}}\right],
\end{equation}
which is also the same as the left-hand side.
\end{proof}

Although $f(x)$ is exactly representable as a linear combination of Chebyshev polynomials of order at most $2b-1$, observe that the coefficients corresponding to the higher order terms are very small. Thus we can truncate the series expansion for $f(x)$ while still remaining $\epsilon$-close to $f(x)$.

\begin{lemma}
\label{lem:approx}
The function $f(x)$ in \eq{taming} can be $\epsilon$-approximated by a linear combination of Chebyshev polynomials of order $O(\sqrt{b\log(b/\epsilon)})$ by truncating the series in \eq{exact} at $j_0 = \sqrt{b\log(4b/\epsilon)}$. In other words,
\begin{equation}
  g(x) := 4 \sum_{j=0}^{j_0} (-1)^j \left[\frac{\sum_{i=j+1}^b \binom{2b}{b+i}}{2^{2b}}\right] \chebt_{2j+1}(x)
\label{eq:approx}
\end{equation}
is $\epsilon$-close to $f(x)$ on $[-1,1]$.
\end{lemma}

\begin{proof}
The expression in square brackets in \eq{exact} is the probability of seeing more than $b+j$ heads on flipping $2b$ fair coins, which is negligible for large $j$.  In particular, the Chernoff bound gives 
\begin{align}
	\frac{1}{2^{2b}} \sum_{i=j+1}^b \binom{2b}{b+i} \le e^{-j^2/b}.
\end{align}
Thus we have
\begin{align}
	|f(x)-g(x)|
	&= 4 \left| \sum_{j=j_0+1}^{b-1} (-1)^j \left[\frac{\sum_{i=j+1}^b \binom{2b}{b+i}}{2^{2b}}\right] \chebt_{2j+1}(x) \right| \\
	&\le 4 \sum_{j=j_0+1}^{b-1} e^{-j^2/b} |\chebt_{2j+1}(x)| 
	\le 4b e^{-j_0^2/b} 
	= \epsilon,
\end{align}
where we used the fact that $|\chebt_n(x)|\le 1$ for $x \in [-1,1]$.
\end{proof}

We now establish the query complexity part of \thm{Chebyshev}.

\begin{proof}[Proof of \thm{Chebyshev} (query complexity)]
We implement $A^{-1}$ using \cor{lcufunc} by expressing it in terms of $g(H)$. Since $\norm{H^{-1} - g(H)} \leq \epsilon$, $\norm{A^{-1} - \frac{1}{d}g(H)} \leq \epsilon/d \leq \epsilon$. \cor{lcufunc} makes $O(\alpha)$ uses of $U$, $V$, and $\boxB$. The value of $\alpha$ for the linear combination $\frac{1}{d}g(H)$ is 
\begin{align}
  \alpha =
  \frac{4}{d} \sum_{j=0}^{j_0} \left[\frac{\sum_{i=j+1}^b \binom{2b}{b+i}}{2^{2b}}\right]
  &\le \frac{4 j_0}{d}
  \label{eq:chebyl1norm}
\end{align}
because the term in square brackets is a probability and hence is at most 1.
Since $V$ costs no queries to implement, it suffices to understand the cost of implementing $U$.
 
The highest order of a Chebyshev polynomial used in \eq{approx} is $O(j_0) = O(d \kappa \log(d\kappa/\epsilon))$. Since the cost of implementing the map \eq{ChebH} is proportional to $n$ as discussed in \sec{Chebyshev_atoms}, the cost of implementing the unitary $U$ in \lem{LCU} is $O(j_0)$.  Thus the total query complexity is $O(\alpha j_0) = O(d\kappa^2 \log^{2}(d\kappa/\epsilon))$ 
Finally, the total number of applications of $\boxB$ is $O(\alpha) = O(\kappa \log(\kappa d/\epsilon))$.
\end{proof}

\subsection{Gate complexity}
\label{sec:Chebyshev_gate}

To upper bound the gate complexity of this approach, we must consider the implementation of the operator $U$ (which consists of steps of the walk $W$) and the operator $V$ in \cor{lcufunc}.

\begin{proof}[Proof of \thm{Chebyshev} (gate complexity)]
We begin with the gate cost of $W$. As discussed in the proof of \cite[Lemma 10]{BCK15}, a step of the quantum walk can be performed up to error at most $\epsilon'$ with gate complexity $O(\log N +\log^{2.5}(\kappa d/\epsilon'))$.
Thus the gate cost of implementing the map $U$, which involves a power of $W$ as large as $j_0$, is $O(j_0(\log N +\log^{2.5}(\kappa d j_0/\epsilon))) = O(j_0(\log N +\log^{2.5}(\kappa d/\epsilon)))$ using \lem{c-U}.

Implementing the map $V$ involves creating a quantum state in a space of dimension $j_0$. As shown in \cite[Sec IV]{SBM06}, any such state can be created with $O(j_0)$ two-qubit gates. Since the gate cost of $V$ is less than the gate cost of $U$, we can neglect this in our calculations.

Thus we see that the gate complexity exceeds the query complexity by a multiplicative factor of $O(\log N +\log^{2.5}(\kappa d/\epsilon))$, as claimed.
\end{proof}

\section{Improved dependence on condition number}
\label{sec:VTAA}

In this section we reduce the $\kappa$-dependence of both our algorithms from quadratic to nearly linear, establishing \thm{linear}. We use a tool called variable-time amplitude amplification, a low-precision version of phase estimation, and results from \sec{Fourier} or \sec{Chebyshev} (either approach can be used to obtain the results of this section).

To see why this improvement might even be possible, observe that the quadratic dependence on $\kappa$ has two contributions. First, the complexity of performing the linear combination corresponding to $A^{-1}$ (even with very low probability) depends linearly on $\kappa$. More precisely, this refers to the cost of implementing the controlled unitary $U$ in \cor{lcufunc}: in the Fourier approach this cost depends on the length of time for which we simulate a Hamiltonian, and in the Chebyshev approach it depends on the highest-order Chebyshev polynomial used. In both approaches, this cost is close to linear in $\kappa$.
The second contribution comes from amplifying the probability of $A^{-1}$ being successfully applied. As described in \cor{lcufunc}, this contribution is proportional to the quantity $\alpha$, whose dependence on $\kappa$ is also approximately linear in both approaches.

Observe that if we were promised that the state $|b\>$ were only a superposition of eigenvectors of $A$ with singular values that are all close to $1$, the problem would be much easier. As noted in \sec{results}, if we know that $|b\>$ lives in a subspace of $A$ with low condition number $\kappa'$, we can replace $\kappa$ with $\kappa'$ in our upper bounds. If $|b\>$ lives in the subspace of eigenvectors of $A$ with singular values close to $1$, then $\kappa' = O(1)$. 
This suggests that the difficult case is when the state $|b\>$ is a superposition of eigenvectors with singular values close to $1/\kappa$. However, if we were promised that this holds, then we could again improve the complexity of our algorithms. This is because the second contribution in \cor{lcufunc}, proportional to $\alpha$, only arises because we need to amplify the success probability of our implementation of $A^{-1}$. If $|b\>$ were simply a superposition of eigenvectors with singular value close to $1/\kappa$, then this step would be considerably less expensive. Quantitatively, observe that when we apply \cor{lcufunc}, the expected complexity is $O(\alpha/\norm{A^{-1}|b\>})$, which would be $O(\alpha/\kappa)=
\poly(\log(d \kappa/\epsilon))$ in this case, bringing down the cost by a factor of $\kappa$.
Thus at both extremes we can improve the $\kappa$-dependence of our algorithm. The reason our algorithm has nearly quadratic dependence on $\kappa$ is that we need to handle both cases simultaneously. 

In the HHL algorithm, a similar situation also arises. Ambainis exploited this to reduce the complexity of the HHL algorithm using a technique called variable-time amplitude amplification (VTAA)~\cite{Amb12}, a generalization of amplitude amplification. VTAA amplifies the success probability of a quantum algorithm $\A$ with variable stopping times using a sequence of concatenated amplitude amplification steps. If some branches of the computation performed by $\mathcal A$ run longer than others,  VTAA can perform significantly better than standard amplitude amplification.

\subsection{Tools}

We now describe the tools required to construct our algorithm: variable-time amplitude amplification (\thm{VTAA}), a low-precision variant of phase estimation we call gapped phase estimation (\lem{GPE}), and procedures for implementing approximations of $A^{-1}$  that are accurate for different ranges of eigenvalues (\lem{Ainvapprox}). In this section we use the notation $\O(\cdot)$ to ignore factors of $\log(d\kappa/\epsilon)$, i.e., we write  $f(d,\kappa,\epsilon) = \O(g(d,\kappa,\epsilon))$ to indicate $f(d,\kappa,\epsilon) = O(g(d,\kappa,\epsilon)\poly(\log(d\kappa/\epsilon))))$.

\paragraph{Variable-time amplitude amplification} We start by formalizing the notion of a quantum algorithm with variable stopping times. We want to capture the intuitive idea that an algorithm $\A$ has $m$ potential stopping times $t_1,t_2,\ldots,t_m$. At time $t_j$, the algorithm can indicate that it wants to stop by setting the \th{j} qubit of a special clock register to $1$, where the algorithm starts with all clock qubits set to $0$. To enforce that the algorithm actually halts when it indicates that it has, we require that subsequent operations do not affect branches of the computation that have halted in previous steps. We formalize this in the following definition.

\begin{definition}[Variable-time quantum algorithm; cf.\ Section 3.3 of \cite{Amb12}]
Let $\A$ be a quantum algorithm on a space $\H$ that starts in the state $|{0}\>_\H$, the all zeros state in $\H$. We say $\A$ is a \emph{variable-time quantum algorithm} if the following conditions hold:
\begin{enumerate}[itemsep=0pt]
\item $\A$ can be written as the product of $m$ algorithms, $\A = \A_m \A_{m-1} \cdots \A_1$.
\item $\H$ can be written as a product $\H=\H_C\otimes\H_A$, where 
$\H_C$ is a product of $m$ single qubit registers denoted $\H_{C_1}, \H_{C_2}, \ldots, \H_{C_m}$.
\item Each $\A_j$ is a controlled unitary that acts on the registers $\H_{C_j} \otimes \H_A$ controlled on the first 
$j-1$ qubits of $\H_C$ being set to $0$.
\end{enumerate}
\end{definition}

In this definition the $m$ potential stopping times of $\A$ correspond to when the segments $\A_j$ end. The last condition enforces that the algorithm in step $j$ does not perform any operation in branches of the computation that have halted in previous steps. 

Now consider a variable-time quantum algorithm that prepares a state $|\psi_\mathrm{succ}\>$ probabilistically. More precisely, the algorithm has a single-qubit flag register that is measured at the end of the algorithm: if we observe the outcome $1$ on measuring this register, we have successfully prepared the state $|\psi_\mathrm{succ}\>$. Let $p_\mathrm{succ}$ denote the probability of obtaining the desired outcome. If the algorithm's complexity is $t_m$, then simply repeating the algorithm $O(1/p_\mathrm{succ})$ times will yield an algorithm that creates the state $|\psi_\mathrm{succ}\>$ with high probability and that has complexity $O(t_m/p_\mathrm{succ})$. Using standard amplitude amplification, we can do the same with complexity only $O(t_m/\sqrt{p_\mathrm{succ}})$. Variable-time amplitude amplification now allows us to achieve the same with even lower cost if the average stopping time of the algorithm is smaller than its maximum running time, as shown by Ambainis \cite[Theorem 1]{Amb12}. We now specialize the result to use our definition of a variable-time quantum algorithm.

\begin{theorem}[Variable-time amplitude amplification]
\label{thm:VTAA}
Let $\A = \A_m \A_{m-1} \cdots \A_1$ be a variable-time quantum algorithm on the space $\H=\H_C \otimes\H_F \otimes \H_W$. If $|{0}\>_\H$ denotes the all zeros state in $\H$ and $t_j$ denotes the query complexity of the algorithm $\A_j\A_{j-1} \cdots \A_1$, we define
\begin{equation}
p_j = \norm{\Pi_{C_j}\A_j\cdots\A_1|{0}\>_\H}^2 \qquad \mathrm{and} \qquad  t_\mathrm{avg} = \sqrt{\sum_{j=1}^m p_j t_j^2}
\end{equation}
to be the probability of halting at step $j$ and the root-mean-square average query complexity of the algorithm, respectively, 
where $\Pi_{C_j}$ denotes the projector onto $|1\>$ in $\H_{C_j}$.
Additionally, let the success probability of the algorithm and the corresponding output state be denoted
\begin{equation}
p_\mathrm{succ} = \norm{\Pi_F\A_m\cdots\A_1|{0}\>_\H}^2
\quad \mathrm{and} \quad
|\psi_\mathrm{succ}\> = \frac{\Pi_F\A_m\cdots\A_1|{0}\>_\H}
{\norm{\Pi_F\A_m\cdots\A_1|{0}\>_\H}},
\end{equation}
where $\Pi_F$ projects onto $|1\>$ in $\H_F$. Then there exists a quantum algorithm with query complexity
\begin{equation}
O\biggl(\Bigl( t_m + \frac{t_\mathrm{avg}}{\sqrt{p_\mathrm{succ}}} \Bigr) \poly(\log t_m)  \biggr)
\end{equation}
that produces the state $|\psi_\mathrm{succ}\>$ with high probability, and outputs a bit indicating whether it was successful. 
(Here we assume $t_\mathrm{avg}$ is known; alternatively, we can use any known upper bound on $t_\mathrm{avg}$.)
The resulting algorithm is gate-efficient if the algorithm $\A$ is gate-efficient.
\end{theorem}

\paragraph{Gapped phase estimation}
Our algorithm also requires the following simple variant of phase estimation, where the goal is to determine whether an 
eigenphase $\theta \in [-1,1]$ of a unitary $U$ satisfies $0 \le | \theta| \le \varphi$ or  $2\varphi \le | \theta| \le 1$.
The proof is straightforward, but we include it for completeness.

\begin{lemma}[Gapped phase estimation]
\label{lem:GPE}
Let $U$ be a unitary operation with eigenvectors $\ket{\theta}$ satisfying $U\ket{\theta} = e^{  i \theta}\ket{\theta}$, and assume that $\theta \in [-1,1]$.
Let $\varphi \in (0,1/4]$ and let $\epsilon > 0$.
Then there is a unitary 
procedure $\GPE(\varphi,\epsilon)$ making $O(\frac{1}{\varphi}\log \frac{1}{\epsilon})$ queries to $U$ that, on input $\ket 0_C \ket {0}_P \ket{\theta} $, prepares a state $(\beta_0 \ket{0}_C \ket{\gamma_0}_P + \beta_1 \ket{1}_C \ket{\gamma_1}_P) \ket \theta$, where $|\beta_0|^2+|\beta_1|^2=1$, such that
\begin{itemize}[itemsep=0pt,topsep=2pt]
	\item if $0 \le |\theta| \le \varphi$ then $|\beta_1| \le \epsilon $ and
	\item if $2 \varphi \le |\theta| \le 1$ then $|\beta_0| \le \epsilon$.
\end{itemize}
Here $C$ and $P$ are registers of 1 and $\ell$ qubits, respectively, where $\ell = O(\log(1/\varphi) \log(1/\epsilon))$.
\end{lemma}

\begin{proof}
Standard phase estimation can give an estimate of $\theta$ with precision $\varphi/2$ (which suffices to distinguish between the two cases), succeeding with probability greater than $1/2$ \cite[Appendix C]{CEMM98}, using $O(1/\varphi)$ queries to $U$ and $O(\log(1/\varphi))$ ancillary qubits for the estimate of the phase.  
To boost the success probability, we simply repeat the procedure $O(\log (1/{\epsilon}))$ times.
We construct $\GPE(\varphi,\epsilon)$ by taking a majority vote and encoding the result in the register $C$.  By a standard Chernoff bound, this suffices to ensure error probability at most $\epsilon$.
\end{proof}

\paragraph{Approximating $A^{-1}$} Our algorithm $\mathcal A$ also requires a procedure to implement various approximations of $A^{-1}$ that are accurate for different ranges of eigenvectors. If we are promised that the input state $\ket \psi$ is a superposition
of eigenstates of $A$ with eigenvalues in the range $[1,-\lambda] \cup [\lambda,1]$ for some given
$\lambda \in (0,1]$, we can use the results
of \sec{Fourier} or \sec{Chebyshev} to implement a version of $A^{-1}$ accurate for this range of eigenvalues more efficiently than in the general case.

\begin{lemma}
\label{lem:Ainvapprox}
For any $\varepsilon>0$ and $\lambda>0$, there exists a unitary  $W(\lambda,\varepsilon)$ with query complexity  $O(\frac{d}{\lambda}\log^2(\frac{d\kappa}{\varepsilon}))$ satisfying
\begin{align}
W(\lambda,\varepsilon)
\ket 0_F \ket{0}_Q \ket{\psi}_I
=
\frac{1}{\alpha_{\max}} \ket 1_F \ket{0}_Q h(A) \ket{\psi}_I + \ket 0_F \ket{\Psi^\perp}_{QI}
\end{align}
where $\alpha_{\max} = \O(\kappa)$ is a constant independent of $\lambda$,  $\ket{\Psi^\perp}_{QI}$ is an unnormalized quantum state on registers $Q$ and $I$ orthogonal to $|0\>_Q$, and $\norm{h(A)|\psi\>-A^{-1}|\psi\>} \le \varepsilon$ for any $|\psi\>$ that is a superposition of eigenstates of $A$ with eigenvalues in the range $[-1,-\lambda] \cup [\lambda,1]$.  Here $F$ and $I$ are registers of $1$ and $\log_2 N$ qubits, respectively.
\end{lemma}

\begin{proof}
Using either the Fourier or Chebyshev approach, by choosing an appropriately accurate series expansion of $A^{-1}$ we can perform the map
\begin{align}
\label{eq:1/Aoutput}
\ket 0_F \ket{0}_Q \ket{\psi}_I
\mapsto 
\frac 1 \alpha \ket 1_F \ket{0}_Q h(A) \ket{\psi}_I + \ket 0_F \ket{\Psi^\perp}_{QI}.
\end{align}

Using the Fourier approach, the query complexity of the map in \eq{1/Aoutput} is $O((d/\lambda) \log^2(1/\lambda \varepsilon))$.
Using the Chebyshev approach, the query complexity is $O((d /\lambda) \log(d /\lambda \varepsilon))$. In either case, the complexity of the operation in \eq{1/Aoutput} is $O(\frac{d}{\lambda}\log^2(\frac{d\kappa}{\varepsilon}))$.

However, the value of $\alpha$ in \eq{1/Aoutput} will depend on $\lambda$. 
The constant $\alpha$ can be determined from the $L_1$ norm in the approximation of $A^{-1}$ by a sum of unitaries. Under the promise that the input state is a superposition of eigenvectors with eigenvalues in $[-1,-\lambda] \cup [\lambda,1]$, the value of $\alpha$ depends on  $\lambda$ and $\varepsilon$, so we denote it $\alpha=\alpha(\lambda,\varepsilon)$.

Since we require transformations for which the factor $1/\alpha$ appearing in \eq{1/Aoutput} is the same independent of $\lambda$, we choose the largest value of $\alpha$ for our range of eigenvalues, which is $\alpha_{\max} := \alpha(1/\kappa,\varepsilon)$. This value is $\alpha_{\max} = O(\kappa \sqrt{\log(\kappa/ \varepsilon)})$ or $\alpha_{\max} = O(\kappa \log(d \kappa/ \varepsilon))$ for the Fourier or Chebyshev approaches, respectively, which is $\O(\kappa)$ in either case. 

Since $\alpha(\lambda,\varepsilon) \le \alpha_{\max}$, and $\alpha(\lambda,\varepsilon)$ can be computed efficiently, if we implement \eq{1/Aoutput} with a smaller value of $\alpha$ than needed, it is always possible to make the implementation worse and increase the value of $\alpha$. In particular, the algorithm for preparing \eq{1/Aoutput} with $\alpha$ replaced by the fixed value $\alpha_{\max}$ is the algorithm described in \sec{Fourier} or \sec{Chebyshev} (which implements \eq{1/Aoutput}) followed by an additional step that transforms
\begin{align}
	\ket 1_F \mapsto 
	\frac{\alpha(\lambda,\varepsilon)}{\alpha_{\max}} \ket 1_F 
	+\sqrt{1 - \frac{\alpha(\lambda,\varepsilon)^2}{\alpha_{\max}^2}} \ket 0_F,
\end{align}
conditional on $\ket 0_Q$. This gives us a map $W(\lambda,\varepsilon)$ that performs the transformation in \eq{1/Aoutput} with the same value $\alpha_{\max}$ in the denominator independent of $\lambda$, as described in the  lemma.
\end{proof}

\subsection{Algorithm}

We now describe a variable-time quantum algorithm $\A$ built as a sequence of steps $\mA_1,\ldots,\mA_m$ with $m := \ceil{\log_2 \kappa}+1$, so the algorithm is $\mathcal A= \mA_m \cdots \mA_1$.
The algorithm $\A$ uses the following registers:
\begin{itemize}[noitemsep,topsep=2pt] 
	\item an $m$-qubit clock register $C$, labeled  $C_1,\ldots,C_m$, used to determine a region the eigenvalue belongs to (i.e., to store the result of GPE);
	\item a single-qubit flag register $F$ to indicate whether the approximation of $A^{-1}$ was successfully implemented;
	\item a $(\log_2 N)$-qubit input register $I$, initialized to $\ket b$, that finally contains the output state;
	\item a register $P$, divided into registers $P_1, P_2, \ldots, P_m$, to be used as ancilla for GPE; and
	\item a register $Q$ to be used as ancilla in the implementation of $A^{-1}$.
\end{itemize}
The corresponding Hilbert spaces are denoted
$\mathcal H_C$, $\mathcal H_F$, $\mathcal H_I$, $\mathcal H_P$, and $\mathcal H_{Q}$, respectively.
All registers are initialized in $\ket 0$ except for register $I$, which is initialized in $\ket b$.
When we write $\ket 0_X$ we mean that all qubits of register $X$ are in $\ket 0$.

\paragraph{Algorithm $\A_j$} 
We now describe the algorithm $\A_j$, which forms a part of the variable-time algorithm $\A$. In the algorithm below, each call to GPE (\lem{GPE}) uses the unitary operation $U:=e^{i A}$.  Note that this unitary satisfies the assumption of \lem{GPE} as the norm of $A$ is at most 1. 
For all $j \in [m]$, let $\varphi_j := 2^{-j}$, and let  $\varepsilon=\epsilon/(m \alpha_{\max})$. 

Finally, we define $\mA_j$ as the product of the following two unitary operations:
\begin{enumerate}[noitemsep,topsep=2pt]
	\item Conditional on first $j-1$ qubits of $\H_C$ being $|0\>$, apply $\GPE (\varphi_j,\varepsilon)$ on the input state in $I$ using $C_j$ as the output qubit and additional fresh qubits from $P$ as ancilla
	(denoted $P_j$).
	\label{step:GPE}
	\item Conditional on $C_j$ (the outcome of the previous step) being $\ket 1_{C_j}$, apply $W( \varphi_j , m \varepsilon)$ on the input state in $I$ using $F$ as the flag register and register $Q$ as ancilla.
	\label{step:W}
\end{enumerate}

\paragraph{Final algorithm}
Before describing the final algorithm, we define another sequence of algorithms $\A'=\A'_m \cdots \A'_1$, similar to $\A$. The only difference between $\A_j$ and $\A'_j$ is that instead of applying the operator $W$ in \step{W} of  $\A_j$, algorithm $\A'_j$ applies the following trivial operation $W'$:
\begin{align}
\label{eq:triv}
W'\ket 0_F \ket{0}_Q \ket{\psi}_I
=\ket 1_F \ket{0}_Q  \ket{\psi}_I. 
\end{align}
The {final} algorithm of this section is as follows. First, apply VTAA to $\mathcal A$ to produce a normalized version of the state output by $\mathcal A$ that has register $F$ in $\ket 1_F$. Then apply the unitary $(\A')^{\dagger}$ to this state. The resulting state in register $I$ is
$\frac{A^{-1} \ket b}{\norm{A^{-1} \ket b}}$
up to error $O(\epsilon)$, as we prove in the next section.

\subsection{Correctness}

To prove that the algorithm works correctly, it is useful to first analyze its behavior on an eigenstate $|\lambda\>$ of $A$, which satisfies $A|\lambda\> = \lambda |\lambda\>$ for some $\lambda \in [-1,1]$. Then the action of $\mathcal A$ on a general state $|b\>$ follows from linearity. 

Let $j\in[m]$ be  such that $\varphi_j < |\lambda| \leq 2\varphi_j$ (recall that $\varphi_j := 2^{-j}$). Such a $j$ must exist because the largest value of $|\lambda|$ is $1 = 2\varphi_1$, and the smallest value of $|\lambda|$ is $1/\kappa$, which satisfies $\varphi_m = 1/2^{\ceil{\log(\kappa)}+1} < 1/\kappa$. 

{In this section, we will use the notation $|\psi\> = |\phi\> + O(\delta)$ to mean $\norm{|\psi\> - |\phi\>} = O(\delta)$.}

\paragraph{State after $\A_1$ to $\A_{j-1}$} We first observe that the algorithms $\A_1, \ldots, \A_{j-1}$ essentially do nothing to the state except modify the ancilla register $P$. To see this, observe that we start with the state
$|0\>_C |0\>_F |\lambda\>_I |0\>_P |0\>_Q$.
Upon applying \step{GPE} of $\A_1$, because $|\lambda|\leq \phi_1$, the clock register $C_1$ remains $|0\>$ with high probability, and only the register $P_1$ has been modified to hold the ancillary state of GPE. Thus, we have the state
\be 
|0\>_C |0\>_F |\lambda\>_I |\gamma^{1}_0\>_{P_1} |0\>_{P_2 \cdots P_m}|0\>_Q
\ee
up to error $O(\varepsilon)$, where $|\gamma^{1}_0\>$ is the ancillary state produced by GPE. \Step{W} of $\A_1$ does nothing because register $C_1$ is set to $|0\>$. 
Similarly, after $j-1$ steps we are left with the state
\be 
|0\>_C |0\>_F |\lambda\>_I |\gamma^{1}_0\>_{P_1} \cdots |\gamma^{j-1}_0\>_{P_{j-1}}
|0\>_{P_j \cdots P_m}|0\>_Q
\ee
up to error $O(j\varepsilon) = O(m\varepsilon)$, where $|\gamma^{i}_0\>$ is the ancillary state produced by the \th{i} call to GPE.

\paragraph{State after $\A_{j}$} Now when we apply \step{GPE} of $\A_j$, because $|\lambda|$ lies in between $\varphi_j$ and $2\varphi_j$, GPE will return a superposition of $|0\>_{C_j}$ and $|1\>_{C_j}$. We will then have the state
\begin{align}
\label{eq:unfinished}
\beta_0|0\>_C |0\>_F &|\lambda\>_I |\gamma^{1}_0\>_{P_1} \cdots |\gamma^{j-1}_0\>_{P_{j-1}}|\gamma^{j}_0\>_{P_{j}}
|0\>_{P_{j+1}\cdots P_m}|0\>_Q \\ 
+
\beta_1|\mathfrak{U}_j\>_C |0\>_F &|\lambda\>_I |\gamma^{1}_0\>_{P_1} \cdots |\gamma^{j-1}_0\>_{P_{j-1}}|\gamma^{j}_1\>_{P_{j}}
|0\>_{P_{j+1} \cdots P_m}|0\>_Q, \label{eq:finished}
\end{align} 
{up to error $O(m\varepsilon)$,} where $\mathfrak{U}_j:=0^{j-1}10^{m-j}$ denotes the integer $j$ represented in unary. Now when \step{W} of $\A_j$ is applied, the part of the state in \eq{unfinished} remains unchanged, since the register $C_j$ is set to $|0\>$. However, the part in \eq{finished} will have operator $W(\varphi_j,m\varepsilon)$ applied to it and the state in \eq{finished} will become
\begin{align}
\label{eq:fin1}
\beta_1|\mathfrak{U}_j\>_C |1\>_F \Bigl( \frac{h(A)}{\alpha_{\max}}|\lambda\>_I \Bigr)  |\gamma^{1}_0\>_{P_1} \cdots |\gamma^{j-1}_0\>_{P_{j-1}}|\gamma^{j}_1\>_{P_{j}}
|0\>_{P_{j+1}\cdots P_m}|0\>_Q  
+ \beta_1|\mathfrak{U}_j\>_C |0\>_F |g_j\>_{IPQ}, 
\end{align} 
where $|g_j\>_{IPQ}$ is some state on register $I$, $P$, $Q$. Since $\lambda$ falls in the required range of \lem{Ainvapprox}, i.e., it is between $[-1,-\varphi_j] \cup [\varphi_j,1]$, we can replace $h(A)$ in the equation above with $A^{-1}$, incurring error at most $O(m\varepsilon)$.

\paragraph{State after $\A_{j+1}$} We can now apply the algorithm $\A_{j+1}$ to our state. This will only affect the part of the state in \eq{unfinished}, because the computation in \eq{fin1}  has ended and has a nonzero value in register $C$. On applying \step{GPE} of $\A_{j+1}$, the state \eq{unfinished} is mapped to
\begin{align}
\label{eq:unf2}
\beta_0|\mathfrak{U}_{j+1}\>_C |0\>_F |\lambda\>_I & |\gamma^{1}_0\>_{P_1} \cdots |\gamma^{j}_0\>_{P_{j}}|\gamma^{j+1}_1\>_{P_{j+1}}
|0\>_{P_{j+2}\cdots P_m}|0\>_Q
\end{align} 
up to error $O(m\varepsilon)$, since $|\lambda|>2\varphi_{j+1}$, and hence GPE outputs $|1\>$ with high probability. We can now apply \step{W} of $\A_{j+1}$ to this state to obtain
\begin{align}
\beta_0|\mathfrak{U}_{j+1}\>_C |1\>_F \Bigl(\frac{h(A)}{\alpha_{\max}}|\lambda\>_I\Bigr) & |\gamma^{1}_0\>_{P_1} \cdots |\gamma^{j}_0\>_{P_{j}}|\gamma^{j+1}_1\>_{P_{j+1}}
|0\>_{P_{j+2}\cdots P_m}|0\>_Q 
+\beta_0|\mathfrak{U}_{j+1}\>_C |0\>_F |g_{j+1}\>_{IPQ}
\end{align} 
up to error $O(m\varepsilon)$, where $|g_{j+1}\>_{IPQ}$ is some state on register $I$, $P$, $Q$. Since $\lambda$ falls in the required range of \lem{Ainvapprox}, i.e., it is between $[-1,-\varphi_{j+1}] \cup [\varphi_{j+1},1]$, we can replace $h(A)$ in the equation above with $A^{-1}$, incurring error at most $O(m\varepsilon)$.

\paragraph{State after $\A$} Now since the state has no overlap with $|0\>_C$, {up to error $O(m\varepsilon)$,} the remaining operations of $\A$ do nothing. {Thus the final state at the end of algorithm $\A$,  $\A |0\>_{CF} |\lambda\>_I |0\>_{PQ}$, is}
\begin{align}
&\frac{\beta_1}{\alpha_{\max}}|\mathfrak{U}_j\>_C |1\>_F A^{-1}|\lambda\>_I  |\gamma^{1}_0\>_{P_1} \cdots |\gamma^{j-1}_0\>_{P_{j-1}}|\gamma^{j}_1\>_{P_{j}}
|0\>_{P_{j+1}\cdots P_m}|0\>_Q \nonumber\\ 
+&
\frac{\beta_0}{\alpha_{\max}}|\mathfrak{U}_{j+1}\>_C |1\>_F A^{-1}|\lambda\>_I  |\gamma^{1}_0\>_{P_1} \cdots |\gamma^{j}_0\>_{P_{j}}|\gamma^{j+1}_1\>_{P_{j+1}}
|0\>_{P_{j+2}\cdots P_m}|0\>_Q \nonumber\\
+&\beta_1|\mathfrak{U}_j\>_C |0\>_F |g_j\>_{IPQ} 
+\beta_0 |\mathfrak{U}_{j+1}\>_C |0\>_F |g_{j+1}\>_{IPQ}
\end{align} 
up to error $O(m\varepsilon)$. On projecting this state to the $|1\>_F$ subspace, we obtain
\begin{align}
\Pi_F \A |0\>_{CF} |\lambda\>_I |0\>_{PQ} = 
\frac{\beta_1}{\alpha_{\max}}|\mathfrak{U}_j\>_C |1\>_F A^{-1}|\lambda\>_I   |\gamma^{1}_0\>_{P_1} \cdots |\gamma^{j-1}_0\>_{P_{j-1}}|\gamma^{j}_1\>_{P_{j}}
|0\>_{P_{j+1}\cdots P_m}|0\>_Q \nonumber\\ 
+
\frac{\beta_0}{\alpha_{\max}}|\mathfrak{U}_{j+1}\>_C |1\>_F A^{-1}|\lambda\>_I  |\gamma^{1}_0\>_{P_1} \cdots |\gamma^{j-1}_0\>_{P_{j-1}}|\gamma^{j}_0\>_{P_{j}}|\gamma^{j+1}_1\>_{P_{j+1}}
|0\>_{P_{j+2}\cdots P_m}|0\>_Q + O(m\delta),
\end{align} 
where $\Pi_F$ denotes the projector onto $|1\>_F$. Now let $|\Psi_\lambda\>$ denote a normalized quantum state on registers $C$, $F$, $P$, and $Q$, such that this equation can be rewritten as
\be
\label{eq:finstate}
\Pi_F \A |0\>_{CF} |\lambda\>_I |0\>_{PQ} = \frac{A^{-1}}{\alpha_{\max}} |\lambda\>_I |\Psi_\lambda\>_{CFPQ}  + O(m\delta).
\ee 

\paragraph{Algorithm $\A'$} Now observe that if we had run algorithm $\A'$ instead of algorithm $\A$, the ancillary state $|\Psi_\lambda\>$ created would be identical since this state depends only on GPE and not on the implementation of $A^{-1}$. In $\A'$, since the operation $W$ is replaced by the operation in \eq{triv} that simply flips the bit in register $F$, we have 
\be
\A' |0\>_{CF} |\lambda\>_I |0\>_{PQ} = |\lambda\>_I |\Psi_\lambda\>_{CFPQ}  + O(m\delta).
\ee
Note that there is no need to project onto $|1\>_F$, since the final state here has no overlap on $|0\>_F$, and hence there is no need to apply VTAA either. Thus the inverse of algorithm $\A'$ can be used to erase the state $|\Psi_\lambda\>$ given state $|\lambda\>$ in another register, even in superposition.

\paragraph{Final algorithm}
We can now analyze the general case where the input state is $\ket{b} = \sum_k c_k \ket{\lambda_k}$,
where $\ket{\lambda_k}$ are eigenvectors of $A$ of eigenvalue $\lambda_k \in [-1,1]$ and $\sum_k |c_k|^2=1$.

Using linearity and \eq{finstate}, we obtain
\be
\label{eq:finalstate}
\Pi_F \A |0\>_{CF} |b\>_I |0\>_{PQ} 
= \frac{A^{-1}}{\alpha_{\max}} \sum_k c_k |\lambda_k\>_I |\Psi_{\lambda_k}\>_{CFPQ}  + O(m\delta).
\ee 

We can now apply VTAA (\thm{VTAA}) to $\mathcal A$ to produce a normalized version of this state with constant probability. 
The approximation error is $O(\alpha_{\max} m \varepsilon) = O(\epsilon)$ since, in the worst case, the norm of this state is $O(1/\alpha_{\max})$. Now applying $(\A')^\dagger$ allows us to erase the ancillary states $|\Psi_{\lambda_k}\>_{CFPQ}$ to obtain 
\be
\frac{A^{-1}|b\>_I |0\>_{CFPQ}}{\norm{A^{-1}|b\>_I |0\>_{CFPQ}}}
\ee 
up to error $O(\epsilon)$, as desired.

\subsection{Complexity}

The complexity of our algorithm is dominated by the cost of applying VTAA to algorithm $\A$. The next step, which is to apply $(\A')^\dagger$, can only cost as much as $\A$, and hence can be ignored.

To obtain the query complexity of applying VTAA to $\A$, we need to compute the quantities in \thm{VTAA}, which are $t_m$, $p_\mathrm{succ}$, and $t_\mathrm{avg}$.

\paragraph{Computing $t_j$}
Let us begin with  $t_j$, the query complexity of applying $\A_j \cdots \A_1$. The cost of any $\A_j$ is the cost of applying $\GPE(\varphi_j,\varepsilon)$ using \lem{GPE} and $W(\varphi_j,m\delta)$ using \lem{Ainvapprox}. 
The number of uses of $U=e^{iA}$ in $\GPE(\varphi_j,\varepsilon)$ follows from \lem{GPE} and is $O(2^j \log(1/\varepsilon))$.
Then, each $U$ must be implemented using a Hamiltonian simulation algorithm within precision $O(\varepsilon /(2^j \log(1/\varepsilon)))$.
By the results of \cite{BCK15}, the query complexity of $\GPE(\varphi_j,\varepsilon)$ is therefore
\begin{align}
  O(d 2^j \poly(\log(d 2^j/\varepsilon))) = \O(d2^j).
 \end{align}
The query complexity of $W(\varphi_j,m\varepsilon)$ is $O(\frac{d}{\varphi_j}\log^2(\frac{d\kappa}{m\varepsilon}))=  \O(d2^j)$. Since the cost of $\A_j$ is $\O(d2^j)$, $t_j = \O(d2^j)$ and $t_m = \O(d\kappa)$.  

\paragraph{Computing $p_\mathrm{succ}$}

The probability $p_\mathrm{succ}$ is the probability of measuring the register $F$ in $\ket 1_F$ in the state output by algorithm $\mathcal A$. This is simply the squared norm of the state on the left hand side of \eq{finalstate}, i.e., $p_\mathrm{succ}=\norm{\Pi_F \A |0\>_{CF} |b\>_I |0\>_{PQ}}^2$. Hence from \eq{finalstate} we have
\begin{align}
\sqrt{p_\mathrm{succ}} & 
= \Norm{\frac{A^{-1}}{\alpha_{\max}} \sum_k c_k |\lambda_k\>_I |\Psi_\lambda\>_{CFPQ}}  + O(m\delta)\\
&= \frac{1}{\alpha_{\max}} \Norm{\sum_k \frac{c_k}{\lambda_k} |\lambda_k\>_I |\Psi_\lambda\>_{CFPQ}} + O(m\delta)\\
&= \frac{1}{\alpha_{\max}} {\left(\sum_k \frac{|c_k|^2}{\lambda_k^2}\right)^{\frac{1}{2}}} + O(m\delta). 
\end{align}
To be precise, this is only an approximation of $\sqrt{p_\mathrm{succ}}$ up to error $O(m \varepsilon) = O(\epsilon/\alpha_{\max})$. However, since this error is much smaller than the calculated value of $\sqrt{p_\mathrm{succ}}$, which is at least $1/\alpha_{\max}$, this error only affects the value by a constant factor for sufficiently small $\epsilon$.

\paragraph{Computing $t_\mathrm{avg}$}

Let $p_j$ denote the probability that algorithm $\mathcal A$ stops exactly at the \th{j} step, which is 
$\norm{\Pi_{C_j}\A_j\cdots\A_1|b\>_I|0\>_{CFPQ}}^2$, where $\Pi_{C_j}$ is the projector onto $|1\>_{C_j}$. We can now compute
\begin{align}
t_\mathrm{avg}^2 = \sum_j p_j t_j^2 &= \sum_j \norm{\Pi_{C_j}\A_j\cdots\A_1|b\>_I|0\>_{CFPQ}}^2 t_j^2 \\
&=\sum_j \norm{\Pi_{C_j}\A_j\cdots\A_1\sum_k c_k|\lambda_k\>_I|0\>_{CFPQ}}^2 t_j^2 \\
&=\sum_j \sum_k |c_k|^2 \norm{\Pi_{C_j}\A_j\cdots\A_1 |\lambda_k\>_I|0\>_{CFPQ}}^2 t_j^2 \\
&=\sum_k |c_k|^2 \Bigl(\sum_j \norm{\Pi_{C_j}\A_j\cdots\A_1 |\lambda_k\>_I|0\>_{CFPQ}}^2 t_j^2\Bigr), \label{eq:tavg}
\end{align}
where the expression in parentheses is the average squared stopping time of the algorithm for the special case where $|b\>= |\lambda_k\>$. For the state $|\lambda_k\>$, we know that the algorithm stops after step $j$ or $j+1$, where $j$ satisfies $\varphi_j < |\lambda_k| \leq 2\varphi_j$. Hence the average squared stopping time for $|\lambda_k\>$ is at most $t_{j+1} = \O(d2^j) = \O(d/\lambda_k)$. More precisely, the algorithm stops after step $j$ or $j+1$ with probability at least $1-\delta$. With probability at most $\delta$ it may stop at step $j+2$, and with probability at most $\delta^2$ it may stop at step $j+3$, etc. Since $t_{j+r}$ grows only exponentially with $r$, i.e., $t_{j+r} = \exp(r) t_j$, for small enough $\delta$, the sum $\delta^r \exp(r)$ converges to a constant and hence this does not change the average squared stopping time for $|\lambda_k\>$ by more than a constant factor. 

Substituting this into \eq{tavg}, we have
\be 
t_\mathrm{avg}^2 = \sum_k |c_k|^2 \O(d^2/\lambda_k^2)=\O\biggl(d^2\sum_k \frac{|c_k|^2}{\lambda_k^2}\biggr).
\ee 

\paragraph{Total complexity}

We are now ready to compute the total cost of VTAA and prove \thm{linear}. \thm{VTAA} states that the query complexity of VTAA applied to $\mathcal A$ is
\begin{equation}
O\(\Bigl( t_m + \frac{t_\mathrm{avg}}{\sqrt{p_\mathrm{succ}}} \Bigr) \poly(\log t_m)  \).
\end{equation}
Using the values we computed, this is 
\begin{align}
\O(d \kappa + d \alpha_{\max}) =  \O(d\kappa),
\end{align}
where we used the fact that $\alpha_{\max} = \O(\kappa)$, as stated in \lem{Ainvapprox}.

Finally, since $\A$ and $\A'$ are gate-efficient and VTAA preserves this property, the overall algorithm is gate-efficient.

\section*{Acknowledgements}

RK thanks Vamsi Pritham Pingali for several helpful discussions during the course of this work. We would also like to thank the anonymous referees of SICOMP for their detailed comments.

The authors acknowledge support from
AFOSR grant number FA9550-12-1-0057,
ARO grant numbers W911NF-12-1-0482 and W911NF-12-1-0486,
CIFAR,
IARPA grant number D15PC00242,
NRO,
and NSF grant number 1526380.
This preprint is MIT-CTP \#4687 and LA-UR-15-27205.


\newcommand{\etalchar}[1]{$^{#1}$}


\end{document}